%% file: 01_main.tex
\documentclass{article}

\usepackage{arxiv}

\usepackage[utf8]{inputenc} 
\usepackage[T1]{fontenc}    
\usepackage{hyperref}       
\usepackage{url}            
\usepackage{booktabs}       
\usepackage{amsfonts}       
\usepackage{amsmath}
\usepackage{nicefrac}       
\usepackage{microtype}      
\usepackage{lipsum}
\usepackage{graphicx}
\usepackage{bm}
\usepackage{natbib}
\usepackage{float}
\usepackage{enumitem}
\usepackage{amsthm}

\newtheorem*{theorem}{Theorem}
\newtheorem*{corollary}{Corollary}

\newcommand*{\N}{\mathop{\mathrm{N}}\nolimits}

\newcommand*{\W}{\mathop{\mathrm{W}}\nolimits}

\newcommand*{\E}{\mathop{\mathrm{E}}\nolimits}

\newcommand*{\Cov}{\mathop{\mathrm{Cov}}\nolimits}
\newcommand*{\Corr}{\mathop{\mathrm{Corr}}\nolimits}
\newcommand*{\blkdiag}{\mathop{\mathrm{blkdiag}}\nolimits}
\newcommand*{\tr}{\mathop{\mathrm{tr}}\nolimits}
\newcommand*{\vect}{\mathop{\mathrm{vec}}\nolimits}
\newcommand*{\POS}{\mathop{\mathrm{POS}}\nolimits}

\title{Bayesian Multivariate Probability of Success Using Historical Data with Strict Control of Family-wise Error Rate}

\author{
 Ethan M. Alt \\
  Department of Biostatistics\\
  University of North Carolina\\
  Chapel Hill, NC 27599 \\
  \texttt{ethanalt@live.unc.edu} \\
   \And
 Matthew A. Psioda \\
  Department of Biostatistics\\
  University of North Carolina\\
  Chapel Hill, NC 27599 \\
  \texttt{matt\_psioda@unc.edu} \\
  \And
 Joseph G. Ibrahim \\
  Department of Biostatistics\\
  University of North Carolina\\
  Chapel Hill, NC 27599 \\
  \texttt{ibrahim@bios.unc.edu} \\
}

\begin{document}
\maketitle
\begin{abstract}
Given the cost and duration of phase III and phase IV clinical trials, the development of statistical methods for go/no-go decisions is vital. In this paper, we introduce a Bayesian methodology to compute the probability of success based on the current data of a treatment regimen for the multivariate linear model. Our approach utilizes a Bayesian seemingly unrelated regression model, which allows for multiple endpoints to be modeled jointly even if the covariates between the endpoints are different. Correlations between endpoints are explicitly modeled. This Bayesian joint modeling approach unifies single and multiple testing procedures under a single framework. We develop an approach to multiple testing that asymptotically guarantees strict family-wise error rate control, and is more powerful than frequentist approaches to multiplicity. The method effectively yields those of Ibrahim et al. and Chuang-Stein as special cases, and, to our knowledge, is the only method that allows for robust sample size determination for multiple endpoints and/or hypotheses and the only method that provides strict family-wise type I error control in the presence of multiplicity.
\end{abstract}



\maketitle

\section{Introduction}
\label{sec1}

Clinical trials fail when there is a lack of statistical evidence that treatment is efficacious. One of the reasons for this is because the sample size was too small to establish statistical significance in the treatment effect when, in actuality, treatment is efficacious, i.e., the study is \emph{underpowered}. Sample size determination becomes more difficult in the presence of multiple hypotheses (\emph{e.g.}, co-primary endpoints, multiple primary endpoints,  or primary and secondary endpoints), which are common in clinical trials. Sample sizes are frequently determined based solely on establishing efficacy in the primary endpoint(s). However, if multiple hypotheses are being considered, it is recommended to determine a sample size that is likely to provide substantial evidence of efficacy across all hypotheses \citep{fda_multiple_2018}. Unfortunately, this is seldom done in practice, in part due to a lack of statistical tools that can be used to robustly determine sample size in the presence of multiple, potentially correlated endpoints.

Due to safety concerns, potential lack of efficacy, cost, and the high failure rate for Phase III clinical trials, it has become common for clinical trials sponsors to make go/no-go decisions on the basis of data from phase II trials quantify the likelihood of a successful Phase III clinical trial. Typically, decisions involved in study planning are made on the basis of classical hypothesis testing, controlling for type-I and type-II error rates by defining a pre-specified criteria for success or futility of the treatment. Sample sizes are often calculated based on an assumed treatment effect size and its variance, with success being confirmed if a significant $p$-value (say, $p$-value $< 0.05$) is obtained. However, variabilities in the treatment effect size and its variance are unincorporated in this traditional approach. 

As discussed by 
\cite{ohagan_assurance_2005}, 
\cite{chuang-stein_sample_2006}, and 
\cite{ibrahim_bayesian_2015},
probability of success (POS) can be interpreted as a weighted average of power across possible effect sizes, i.e., POS is the expectation of power with respect to some predictive distribution for future data. 
\cite{ohagan_assurance_2005} referred to POS as \emph{assurance}, and it is also known as \emph{Bayesian expected power} or simply \emph{Bayesian power} \citep{psioda_bayesian_2019}.

The literature for computing POS for the no-covariate setting is vast, with most methods using Bayesian techniques to sample future clinical data while using frequentist approaches to determine success for the simulated data sets. \citet{wang_evaluating_2013}, studying oncology trials, proposed to use the posterior distribution of the treatment effect to generate data for future trials, computing whether or not the trial had a significant $p$-value. \citet{ren_assurance_2014} developed methods to compute POS for a single time-to-event endpoint for both parametric and nonparametric models. \citet{ciarleglio_selection_2016} and \citet{ciarleglio_sample_2017} utilized a hybrid no-covariates method to compute POS using Bayesian conditional expected power (BCEP), which corresponds to POS conditioning on a non-null treatment effect, for superiority trials for normally distributed and binary endpoints, respectively.

In contrast, methods for computing POS in the presence of covariates are sparse. \citet{ibrahim_bayesian_2015} developed a fully Bayesian approach to compute POS for a single normally distributed endpoint in the presence of covariates. \citet{shieh_equivalence_2017} showed that, under certain conditions, so-called \emph{multiplier methods}, where multiples of the sample variance are used determine sample size, are equivalent to POS for the normal linear model.

Recently, some methods have been developed for computing POS for correlated outcomes or when multiple trials must be conducted (e.g., multiple Phase III or Phase II and III trials). \citet{zhang_joint_2013} developed a method to compute POS when multiple successful trials are required, noting that two successful phase III trials are needed for approval and that computing POS using common observed data from earlier studies yields correlated posterior treatment effects. \citet{yin_backward_2017} developed a Bayesian approach for determining go / no-go decisions for phase 2 clinical trials based on phase 3 efficacy and probability of success requirements. \citet{zhou_predictive_2018} developed closed-form equations for POS methods for interim monitoring in clinical trials with longitudinal outcomes. \citet{saint-hilary_predictive_2019} develop a method to compute probability of success using surrogate endpoints, which are common in oncology trials. However, none of these methods consider regression settings.

In this paper, we provide a method to compute POS for multiple continuous endpoints, where the regression model for each endpoint has its own set of covariates and the dependent variables can reasonably be modeled by a multivariate normal distribution. The method proposed includes those of \cite{ibrahim_bayesian_2015} and \cite{chuang-stein_sample_2006} as special cases. This fully Bayesian approach does not rely on $p$-values, hybrid approaches (e.g., \citet{chuang-stein_sample_2006}), or hierarchical testing procedures (e.g. \citet{holm_simple_1979}). The proposed Bayesian approach allows treatment effects to be sampled from the joint posterior density, which will typically be correlated due to the related nature of endpoints commonly used in clinical trials. We develop a Bayesian procedure to compute probability of success for unions of events that asymptotically guarantees strict family-wise error rate (FWER) control. While no closed forms are available, sampling maybe conducted via the direct Monte Carlo (DMC) method of \citet{zellner_direct_2010} under some circumstances and an efficient Gibbs sampling algorithm in others. 

The proposed Bayesian approach to testing has several advantages over frequentist approaches to multiplicity. First, since FWER is strictly controlled, power loss in the presence of multiplicity is not a concern. Typically, frequentist approaches for multiplicity adjustment, such as Bonferroni and Holm, are conservative, that is, they typically yield FWER less than $\alpha$, necessarily translating to decreased power. Other methods to handle multiplicity, such as the fixed sequence testing and Hochberg procedures, can be more powerful than those of Bonferroni and Holm. However, the former method requires a priori knowledge about the level of statistical significance of the endpoints and the latter requires the assumption that the tests are positively correlated, and they are not recommended by regulators \citep{fda_multiple_2018}. By contrast, our proposed method is invariant to the order of testing and correlations. Furthermore, in clinical settings, composite endpoints may be created to avoid frequentist multiplicity corrections due to the loss in power \citep{fda_multiple_2018}. The proposed approach allows clinicians to declare efficacy if the criteria at least one of several endpoints is met, while guaranteeing FWER of $\alpha$.

The rest of the paper is organized as follows. In Section 2, we introduce the Bayesian POS and develop the framework to compute POS for multiple continuous endpoints. Section 3 introduces a posterior probability adjustment method that asymptotically guarantees strict FWER of level $\alpha$ in the presence of union hypotheses. Section 4 presents simulation results, comparing the proposed method with traditional frequentist approaches. Section 5 presents results from an application of the proposed method using data from a real clinical trial. In Section 6, we close with some discussion.

\section{Bayesian probability of success for multivariate linear regression models}
\label{sec:posmethod}
\subsection{The general methodology}
Let $z_i$ denote the treatment indicator, such that $z_i = 1$ if the $i^{th}$ patient receives treatment and $z_i = 0$ otherwise. Further, let $\bm{x}_{2ij}$ denote a $p_j$-dimensional covariate vector for the $i^{th}$ individual and $j^{th}$ endpoint, $i = 1, \ldots, n$ and $j = 1, \ldots, J$, which may include an intercept term. Let the $j^{th}$ outcome variable for subject $i$ be represented by $y_{ij}$. We consider a two-arm a randomized controlled trial (RCT) and let
$D = \{ \left( y_{ij}, z_i, \bm{x}_{2ij} \right), i = 1, \ldots, n; j = 1, \ldots, J \}$ represent the future study data. We assume a linear regression model of the form
\begin{equation}
    y_{ij} = \beta_{1j} z_{i} + \bm{x}_{2ij}'\bm{\beta}_{2j} + \epsilon_{ij}, \hspace{1cm} i = 1, \ldots, n; j = 1, \ldots, J,
    \label{eq:indivreg}
\end{equation}
where 
$\beta_{1j}$ is the treatment effect for outcome $j$,
$\bm{\beta}_{2j}$ is a $p_j \times 1$ parameter vector corresponding to $\bm{x}_{2ij}$, and 
$\epsilon_{ij}$ is a random error term. 
We assume that $\bm{\epsilon}_i = (\epsilon_{i1}, \ldots, \epsilon_{iJ})' \sim \N_J(\bm{0}_J, \bm{\Sigma})$, where $\bm{0}_J$ is a $J$-dimensional vector of zeros and $\bm{\Sigma}$ is a covariance matrix assumed to be positive definite. We see from the model (\ref{eq:indivreg}) that each individual outcome $y_{ij}$ marginally has a linear regression model, but the $J$ outcomes are correlated with one another through the covariance matrix $\bm{\Sigma}$. This model is referred to in the literature as a seemingly unrelated regression (SUR) model \citep{zellner_direct_2010}, where the name derives from the fact that each marginal regression model is valid and potentially has its own set of covariates, but the outcomes are correlated. The SUR model is particularly useful for clinical trials, where multiple outcomes are obtained from the same experimental units (i.e. the trial participants) and hence are correlated, but where one might wish to control only for the baseline level of the endpoint being considered, so that the covariate vector for each outcome is different. Moreover, the Bayesian SUR model allows for sampling from the \emph{joint} posterior distribution of the regression coefficients, which allows joint inference on the treatment effects.

Let $\bm{y}_j = (y_{i1}, \ldots, y_{iJ})'$ denote the responses for outcome $j$. Let $\bm{X}_j = (\bm{x}_{1j}, \ldots, \bm{x}_{nj})'$ denote the $n \times p_j$ design matrix for outcome $j$ and let $\bm{z}_i = (z_{i1}, \ldots, z_{in})'$ denote the $n \times 1$ vector of treatment indicators. Then we can rewrite (\ref{eq:indivreg}) as
\[
    \bm{y}_j = \beta_{1j} \bm{z}_i + \bm{X}_{2j} \bm{\beta}_{2j} + \bm{\epsilon}_j, \hspace{1cm} j = 1, \ldots, J,
    \label{eq:outcomereg}
\]
where the $\bm{\epsilon}_j$'s are zero-mean, normally distributed error terms with $\Cov(\bm{\epsilon}_j, \bm{\epsilon}_{k}) = \sigma_{jk} \bm{I}_n$, where $\sigma_{jk}$ is the element in the $j^{th}$ row and $k^{th}$ column of the covariance matrix $\bm{\Sigma}$, and $\bm{I}_n$ is the $n$-dimensional identity matrix.

We can write the linear regression model even more succinctly. Let $\bm{y} = (\bm{y}_1, \ldots, \bm{y}_J)'$ be the $nJ \times 1$ vector of all responses. Further, let $\bm{X}_{2} = \blkdiag\{ \bm{X}_{21}, \ldots, \bm{X}_{2J} \}$ be the $nJ \times p$ block diagonal matrix of design matrices and let $\bm{Z} = \blkdiag\{\bm{z}, \ldots, \bm{z}\}$ be the $nJ \times J$ block diagonal matrix of treatment indicators, where $p = \sum_{j = 1}^J p_j$. Let $\bm{\beta}_1 = ({\beta}_{11}, \ldots, {\beta}_{1J})'$ and $\bm{\beta}_2 = (\bm{\beta}_{21}', \ldots, \bm{\beta}_{2J}')'$. We can write
\begin{equation}
    \bm{y} = \bm{Z}\bm{\beta}_1 + \bm{X}_2 \bm{\beta}_2 + \bm{\epsilon},
    \label{eq:jointregtrt}
\end{equation}
where $\bm{\epsilon} = (\bm{\epsilon}_1, \ldots, \bm{\epsilon}_J)'$ are the stacked error terms and $\Cov(\bm{\epsilon}) = \bm{\Sigma} \otimes \bm{I}_n$, where $\otimes$ is the Kronecker product operator. Let $\bm{\beta} = (\bm{\beta}_1^*, \ldots, \bm{\beta}_J^*)'$, where 
$\bm{\beta}_j^* = (\beta_{1j}, \bm{\beta}_{2j})'$ contains all regression parameters for outcome $j$. Let $\bm{X} = \blkdiag\{ \bm{X}_1^*, \ldots, \bm{X}_J^* \}$, where $\bm{X}_j^* = (\bm{z}, \bm{X}_{2j})$ is the $n \times (p_j + 1)$ design matrix for outcome $j$. Then we can write (\ref{eq:jointregtrt}) as
\[
    \bm{y} = \bm{X}\bm{\beta} + \bm{\epsilon}.
\]

For univariate endpoints, ``success'' is typically defined in terms of the treatment effect exceeding a certain threshold \citep{ibrahim_bayesian_2015}. However, a more general definition of success is required for multiple endpoints. The null hypothesis can be expressed as 
$H_0: \bm{\beta}_1 \in \bar{\Omega}$ versus 
$H_1: \bm{\beta}_1 \in \Omega$,
where $\Omega$ is a set that defines success and $\bar{\Omega}$ is the complement of $\Omega$. 
For example, consider $J = 1$, then success is achieved if $\beta_1 > \delta$, where $\delta \in \mathbb{R}$ is a target value. Here, $\Omega$ is just an interval, $\Omega = (\delta, \infty)$. We note that for $J > 1$, $\Omega$ can be quite complex. If $J = 2$ for example, we might define success if $\beta_{11} > \delta_1$, $\beta_{12} > \delta_2$, or $a_1 \beta_{11} + a_2 \beta_{12} > \delta_3$ for specified constants $a_1$ and $a_2$. For clinical trials with 3 or more endpoints, it is often desirable to establish efficacy in the primary endpoint and at least one of a collection of secondary endpoints. In this case, we have 
$\Omega = \{ \beta_{11} > \delta_1 \} \cap \left(\cup_{j = 2}^J \{ \beta_{1j} > \delta_j \} \right) $, where we assume without loss of generality that the first outcome is primary and the others are secondary.

Mathematically, we can define probability of success as
\begin{equation}
    \POS = \int \int
      P(\text{success} | z, \bm{x}, \bm{\theta}) 
      f(z)
      f(\bm{x} | \bm{\alpha})
      \pi^{(v)}(\bm{\theta}, \bm{\alpha}) 
      dz 
      d\bm{x} 
      d\bm{\theta} 
      d{\bm{\alpha}},
      \label{eq:pos}
\end{equation}
where $\bm{\theta} = (\bm{\beta}, \bm{\Sigma})$, $\bm{\alpha}$ is a vector of parameters for the covariate distribution, $f(z)$ is a known distribution determined by the randomization scheme, $f(\bm{x} | \bm{\alpha})$ is the covariate distribution, and $\pi^{(v)}(\bm{\theta}, \bm{\alpha})$ is a validation prior for $\bm{\theta}$ and $\bm{\alpha}$. For the purposes of this paper, the validation prior represents the posterior density with respect to a historical data set (e.g., phase II trial data). In the literature, the validation prior is also referred to as the sampling prior \citep{wang_simulation-based_2002} and the design prior \citep{ohagan_assurance_2005}. Moreover, we define
\begin{equation}
    P(\text{success} | z, \bm{x}, \bm{\theta}) =
      \E\left[1\left\{ P( \bm{\beta}_1 \in \Omega | D, \pi^{(f)} ) \ge \gamma \right\} | z, \bm{x}, \bm{\theta} \right],
      \label{eq:success}
\end{equation}
where $D = \{ \bm{y}, \bm{z}, x \}$ represents the future data and $P(\bm{\beta}_1 \in \Omega | D, \pi^{(f)})$ is the posterior probability that $\bm{\beta}_1$ lies within the region of success given the data from the future trial and the fitting prior $\pi^{(f)}$. The fitting prior is called so because it is the prior used to ``fit'' the future data. In words, (\ref{eq:pos}) is the marginal probability that $\bm{\beta}_1$ lies in the success region $\Omega_1$. Following \cite{ibrahim_bayesian_2015}, we assume independent validation priors for $\bm{\theta}$ and $\bm{\alpha}$ in (\ref{eq:pos}). Specifically,
\begin{equation}
    \pi^{(v)}(\bm{\theta}, \bm{\alpha}) = 
      \pi^{(v)}(\bm{\theta}) \pi^{(v)}(\bm{\alpha}).
      \label{eq:jointvalidprior}
\end{equation}
\subsection{Single historical data set}
Suppose that we have historical data 
$D_0 = \{ y_{0ij}, z_{0i}, \bm{x}_{0ij}, i = 1, \ldots, n_0; j = 1, \ldots, J \}$,
where $z_{0i}$ = 1 if the $i^{th}$ subject received treatment and 0 otherwise and
$\bm{x}_{0i}$ is the collection of unique covariates for subject $i$ across all outcomes $j = 1, \ldots, J$. We assume that $\bm{y}_{0} = \bm{X}_0 \bm{\beta} + \bm{\epsilon}$ where, using the notation established above,
$\bm{y}_0 = (\bm{y}_{01}, \ldots, \bm{y}_{0J})'$ is the $n_0J \times 1$ vector of responses, $\bm{X}_{0} = \blkdiag\{ \bm{X}_{01}^*, \ldots, \bm{X}_{0J}^* \}$ is a block diagonal matrix consisting of design matrices for each outcome, and $\bm{\epsilon}_0 \sim \N_{n_0 J}(\bm{0}_{n_0 J}, \bm{\Sigma} \otimes \bm{I}_{n_0})$.

The validation prior for $\bm{\theta}$ can be specified as
\begin{equation}
    \pi^{(v)}(\bm{\theta}) \propto
      \lvert \bm{\Sigma}^{-1} \rvert ^ {n_0/2}
      \exp\left\{-\frac{1}{2}(\bm{y}_0 - \bm{X}_0 \bm{\beta})' (\bm{\Sigma}^{-1} \otimes \bm{I}_{n_0}) (\bm{y}_0 - \bm{X}_0 \bm{\beta}) \right\} \pi_{10}^{(v)}(\bm{\theta})
      \label{eq:sur1},
\end{equation}
which we may alternatively express as
\begin{equation}
      \pi^{(v)}(\bm{\theta}) \propto \lvert \bm{\Sigma}^{-1} \rvert ^ {n_0/2}
         \exp\left\{-\frac{1}{2} \tr\left( \bm{R}_0 \bm{\Sigma}^{-1} \right) \right\}
         \pi_{10}^{(v)}(\bm{\theta}),
         \label{eq:sur2}
\end{equation}
where 
${(\bm{R}_{0})}_{kl} = (\bm{y}_{0k} - \bm{X}_{0k} \bm{\beta}_{k})'(\bm{y}_{0l} - \bm{X}_{0l} \bm{\beta}_{l})$ is a $J \times J$ matrix of dot product residuals,
and $\pi_{10}^{(v)}(\bm{\theta})$ is an initial validation prior for $\bm{\theta}$. To see how (\ref{eq:sur1}) equates to (\ref{eq:sur2}), note that we can write
$(\bm{y}_0 - \bm{X}_0 \bm{\beta})' (\bm{\Sigma} \otimes \bm{I}_{n_0}) (\bm{y}_0 - \bm{X}_0 \bm{\beta}) = \vect(\bm{Y}_0^* - \bm{X}_0^* \bm{B}^*)' (\bm{\Sigma} \otimes \bm{I}_{n_0}) \vect(\bm{Y}_0^* - \bm{X}_0^* \bm{B}^*)$, 
where $\bm{Y}_0^* = (\bm{y}_1, \ldots, \bm{y}_J)$ is the $n_0 \times J$ matrix of outcomes, 
$\bm{X}_0^* = (\bm{X}_{01}^*, \ldots, \bm{X}_{0J}^*)$ is the $n \times (p + J)$ matrix of concatenated design matrices, and $\bm{B}^* = \blkdiag\{\bm{\beta}_1^*, \ldots, \bm{\beta}_J^*\}$ is a $(p + J) \times J$ block diagonal matrix of regression coefficients. Using this relationship and the fact that $\tr(\bm{A}'\bm{BCD}') = \vect(\bm{A})'(\bm{D} \otimes \bm{B}) \vect(\bm{C})$, we have 
$(\bm{y}_0 - \bm{X}_0 \bm{\beta})' (\bm{\Sigma} \otimes \bm{I}_{n_0}) (\bm{y}_0 - \bm{X}_0 \bm{\beta}) = \tr((\bm{Y}_0^* - \bm{X}_0^* \bm{B}^*)' \bm{I}_{n_0} (\bm{Y}_0^* - \bm{X}_0^* \bm{B}^*) (\bm{\Sigma}^{-1})' ) = \tr( \bm{R}_0 \bm{\Sigma}^{-1} )$ since $\bm{R}_0 = (\bm{Y}_0^* - \bm{X}_0^* \bm{B}^*)'(\bm{Y}_0^* - \bm{X}_0^* \bm{B}^*)$.

We may specify an improper prior for $\pi_{10}^{(v)}(\bm{\theta})$ as
\begin{equation}
    \pi^{(v)}_{10}(\bm{\theta}) \propto \lvert \bm{\Sigma}^{-1} \rvert^{-\frac{J + 1}{2}}
    \label{eq:initvalidprior}.
\end{equation}

The validation prior (\ref{eq:sur1}) can be recognized as a posterior from a Bayesian SUR model. For an introduction to Bayesian analysis of the SUR model, we refer the reader to \cite{zellner_direct_2010}.

Following \cite{ibrahim_bayesian_2015}, we elicit the validation prior for $\bm{\alpha}$ via the power prior \citep{ibrahim_power_2015}. Specifically, we specify $\pi_2^{(v)}(\bm{\alpha})$ in (\ref{eq:jointvalidprior}) as
\begin{equation}
    \pi_{2}^{(v)}(\bm{\alpha}) = \pi_2^{(v)}(\bm{\alpha} | D_0, b_0)
       \propto \left[\prod_{i = 1}^{n_0} f(\bm{x}_{0i} | \bm{\alpha}) \right]^{b_0}
            \pi_{20}^{(v)}(\bm{\alpha}),
       \label{eq:alphavalidprior}
\end{equation}
where $0 < b_0 \le 1$ and $\pi_{20}^{(v)}(\bm{\alpha})$ is an initial validation prior. The quantity $b_0$ is a hyperparameter that downweights the likelihood of the historical data. Typically, we choose values of $b_0$ closer to 0 when characteristics of the study participants in the planned future study are substantially different than those in the completed trial and values close to 1 when the characteristics are similar. We assume covariates are generated by a generalized linear model (GLM).

As the covariates are dependent, it may be difficult to specify $f(\bm{x}_{0i} | \bm{\alpha})$. One way to do this is to utilize a factorization of the joint distribution informed by temporal relationships between the variables. For example, suppose we have three covariates $(x_{01}, x_{02}, x_{03})'$, where $x_{01}$ is a binary variable indicating gender, $x_{02}$ is weight measured in kilograms, and $x_{03}$ counts the number of tumors at baseline. It is impossible for weight or number of tumors to cause gender, so we might first generate $x_{01}$. Gender likely has an effect on both weight and number of tumors, but there may not be an obvious way to factorize $x_{02}$ and $x_{03}$. The two permutations may be compared, and a choice may be selected using model selection criteria (e.g., deviance information criterion (DIC) or Bayes factors). Supposing that such analysis yields $x_{02}$, we can generate from the joint distribution via $f(x_{01}, x_{02}, x_{03}) = f(x_{01}) f(x_{02} | x_{01}) f(x_{03} | x_{01}, x_{02})$. 

In general, suppose there are $L$ unique covariates $\{x_{0i1}, \ldots, x_{0iL}, i = 1, \ldots, n_0 \}$. Without loss of generality, suppose that they are ordered so that we want to generate samples of $x_{01}$ first, $x_{02}$ second, and so on, lastly generating samples from $x_{0L}$. Let $\bm{\alpha}_l$ denote the parameter vector for $f(x_{0l} | (x_{01}, \ldots, x_{0,l-1}) )$. We may rewrite the validation prior (\ref{eq:alphavalidprior}) as
\[
    \pi_2^{(v)}(\bm{\alpha} | D_0, b_0) \propto \prod_{l=1}^L \left[ \prod_{i = 1}^{n_0} 
    f(x_{0il} | (x_{0i1}, \ldots, x_{0i,l-1}, \bm{\alpha}_l) \right]^{b_0}
            \pi_{20}^{(v)}(\bm{\alpha}).
\]
Assume without loss of generality that the first $L_1$ covariates are from an exponential family with a dispersion parameter (e.g., Gaussian or gamma) and remaining $L_2 = L - L_1$ covariates have fixed dispersion parameters (e.g., Binomial or Poisson). Then we may specify the initial validation prior as
\begin{equation}
    \pi_{20}^{(v)}(\bm{\alpha}) \propto \prod_{l = 1}^{L_1} \phi_l^{\alpha_l - 1} \exp(-\gamma_l \phi_l),
    \label{eq:alphainitprior}
\end{equation}
where $\phi_l$ is the dispersion parameter for the $l^{th}$ covariate, $l = 1, 2, \ldots, L_1$. The prior (\ref{eq:alphainitprior}) corresponds to a uniform improper prior on the GLM regression coefficients and an a priori independent gamma prior on the dispersion parameters. Such a prior guarantees posterior propriety \citep{ibrahim_bayesian_1991}. For the purposes of our simulations and data analysis, we specify the weakly informative hyperparameters $\alpha_{0l} = \gamma_{0l} = 0.1$.

Since $P(\text{success} | z, \bm{x}, \bm{\theta})$ in (\ref{eq:success}) is independent of $\bm{\alpha}$, the fitting prior only needs to be expressed in terms of $\bm{\theta}$. We take an improper prior
\[
    \pi^{(f)}(\bm{\theta}) \propto \lvert \bm{\Sigma}^{-1} \rvert ^ {- \frac{J+1}{2} }.
\]

Thus, the joint posterior distribution for $\bm{\theta}$ given the future study data $D$ is
\begin{equation}
    \pi(\bm{\theta}| D, \pi^{(f)}) \propto 
     \lvert \bm{\Sigma}^{-1} \rvert^{n/2} 
       \exp\left\{ -\frac{1}{2} 
         (\bm{y} - \bm{X} \bm{\beta})' 
         (\bm{\Sigma}^{-1} \otimes \bm{I}_n)
         (\bm{y} - \bm{X} \bm{\beta})
       \right\}
     \lvert \bm{\Sigma}^{-1} \rvert^{-\frac{J+1}{2}}.
     \label{eq:futurepost}
\end{equation}

\cite{zellner_introduction_1971} notes that the marginal posterior $\pi(\bm{\beta} | D, D_{01}, a_0)$ does not have a known distribution due to the covariate matrices for each outcome $\{\bm{X}_1^*, \ldots, \bm{X}_J^* \}$ possibly being different. However, an efficient Gibbs sampling algorithm may be used to sample from the posterior (\ref{eq:futurepost}).

\subsection{Two historical data sets}
Suppose that we have two historical data sets
$$D_{0k} = \left\{(y_{0kij}, \bm{x}_{0ki}, z_{0ki}), i = 1, \ldots, n_{0k}; j = 1, \ldots, J; k = 1, 2 \right\},$$
where
$z_{0ki} = 1$ if the $i^{th}$ subject in study $k$ received treatment and 0 otherwise,
$\bm{x}_{0ki}$ represents the vector of unique covariates for subject $i$ in study $k$,
$y_{0kij}$ is the response variable for individual $i$ and outcome $j$ in study $k$,
$k = 1,2$, $j = 1, \ldots, J$, $i = 1, \ldots, n_{0k}$.

We assume $\bm{y}_{0k} = \bm{X}_{0k} \bm{\beta} + \bm{\epsilon}$ where 
$\bm{y}_{0k} = (\bm{y}_{0k1}', \ldots, \bm{y}_{0kJ}')'$ is a $nJ$-dimensional response vector of all responses in study $k$,
$\bm{X}_{0k} = \blkdiag\{ \bm{X}_{0k1}^*, \ldots, \bm{X}_{0kJ}^* \}$ is a block diagonal matrix of design matrices for each outcome in study $k$,
and $\bm{\epsilon}_{0k} = (\bm{\epsilon}_{0k1}', \ldots, \bm{\epsilon}_{0kJ})'$ is a $n_{0k}J \times 1$ vector of error terms for each endpoint in study $k$ with $\bm{\epsilon}_{0k} \sim \N_{n_{0k}J}(\bm{0}_{n_{0k}}, \bm{\Sigma} \otimes \bm{I}_{n_{0k}})$

We use the historical data $D_{01}$ via the power prior to elicit the fitting prior
\[
  \pi^{(f)}_{1}(\bm{\theta}) =
  \left[ 
    \lvert \bm{\Sigma} \rvert^{n_{01} / 2} \exp\left\{
      -\frac{1}{2} (\bm{y}_{01} - \bm{X}_{01} \bm{\beta})'
      \left( \bm{\Sigma}^{-1} \otimes {\bm{I}}_{n_{01}} \right)
       (\bm{y}_{01} - \bm{X}_{01} \bm{\beta})
    \right\}
  \right]^{a_0}
  \pi_{10}^{(f)}(\bm{\theta}),
  \label{eq:fitprior}
\]
and we may take $\pi_{10}^{(f)}(\bm{\theta})$ as
\[
    \pi_{10}^{(f)}(\bm{\theta}) \propto \lvert \bm{\Sigma}^{-1} \rvert^{-\frac{J+1}{2}}
    \label{eq:initfitprior}
\]

The joint posterior of the future data is given by
\begin{equation}
    \pi(\bm{\theta} | D, \pi^{(f)}_1) \propto
    \lvert \bm{\Sigma}^{-1} \rvert^{n/2}
    \exp\left\{ -\frac{1}{2}
      (\bm{y} - \bm{X \beta})' (\bm{\Sigma}^{-1} \otimes \bm{I}_n) 
        (\bm{y} - \bm{X \beta})
    \right\}
    \pi_1^{(f)}(\bm{\theta}).
    \label{eq:postfitprior}
\end{equation}

As previously mentioned, the marginal posterior of $\bm{\beta}$ has no closed form. Furthermore, the DMC approach cannot be used here because it requires priors to be noninformative in $\bm{\Sigma}$.  However, the conditional posteriors $\bm{\beta} | D, \bm{\Sigma}^{-1}$ and $\bm{\Sigma}^{-1} | D, \bm{\beta}$ are normal and Wishart, respectively. That is,
\[
    \pi \left( \bm{\beta} | \bm{\Sigma}, D, \pi_1^{(f)} \right) \propto 
    \N_p \left( \bm{\beta} | \tilde{\bm{\beta}}, \tilde{\bm{\Sigma}}_{\bm{\beta}} \right) \label{eq:postfit1}
\]
and
\[
    \pi \left( \bm{\Sigma} | \bm{\beta}, D, \pi_1^{(f)} \right) \propto 
      \W_J \left( \bm{\Sigma}^{-1} | \tilde{\bm{V}}, \tilde{\alpha} \right),
\]
where 
$\N_p(\cdot, \bm{a}, \bm{B})$ denotes the $p$-dimensional normal density with mean $\bm{a}$ and positive-definite covariance matrix $\bm{B}$,
\begin{equation}
{\tilde{\bm{\Sigma}}}_{\bm{\beta}} = 
\left[ %
\bm{X}' \left( \bm{\Sigma}^{-1} \otimes \bm{I}_n \right) \bm{X} + %
a_0^{-1} \bm{X}'_{01} \left( \bm{\Sigma}^{-1} \otimes \bm{I}_{n_{01}} \right) \bm{X}_{01} %
\right]^{-1}
\end{equation}
and
\begin{equation}
\tilde{\bm{\beta}} = \tilde{\bm{\Sigma}}_{\bm{\beta}} 
\left[ 
 \bm{X}'\left( \bm{\Sigma}^{-1} \otimes \bm{I}_n \right) \bm{X} \hat{\bm{\beta}}
+
a_0^{-1} \bm{X}_{01}' \left( \bm{\Sigma}^{-1} \otimes \bm{I}_{n_{01}} \right) \bm{X}_{01} \hat{\bm{\beta}}_{01}
\right],
\end{equation}
where $\W_q(\cdot | \bm{C}, d)$ denotes the $q$-dimensional Wishart probability density function with scale matrix $\bm{C}$ and $d$ degrees of freedom,
$\tilde{\bm{V}} = \tilde{\bm{V}}(\bm{\beta}) = \left[ \bm{R}(\bm{\beta}) + a_0 \bm{R}_0(\bm{\beta}) \right]^{-1}$,
$\tilde{\alpha} = n + a_0 n_0$,
and $\bm{R}$ and $\bm{R}_0$ are the matrices of dot-product residuals for the future data and historical data respectively as described in (\ref{eq:sur2}), and $\hat{\bm{\beta}}$ and $\hat{\bm{\beta}}_{01}$ respectively are the maximum likelihood estimators for $\bm{\beta}$ for the current and historical data, respectively. Thus, Gibbs sampling can be used to sample from the posterior (\ref{eq:postfitprior}). A formal derivation is presented in the Appendix A.

The weight parameter $a_0$ can be treated as either fixed or random. For computational reasons, we consider $a_0$ as fixed. \cite{ibrahim_power_2015} give a through review of the power prior and a discussion of treating $a_0$ as fixed and random.

The validation prior for $\bm{\theta}$ can be specified as
\begin{equation}
    \pi_1^{(v)}(\bm{\theta}) \propto 
      \lvert \bm{\Sigma}^{-1} \rvert ^ {n_{02}/2}
      \exp\left\{-\frac{1}{2}(\bm{y}_{02} - \bm{X}_{02} \bm{\beta})' (\bm{\Sigma}^{-1} \otimes \bm{I}_{n_{02}}) (\bm{y}_{02} - \bm{X}_{02} \bm{\beta}) \right\} \pi_{10}^{(v)}(\bm{\theta}),
      \label{eq:validprior}
\end{equation}
where the initial validation prior is elicited as (\ref{eq:initvalidprior}). The validation prior represents the posterior distribution for $\bm{\theta}$ of the most recently completed study. In principle, the older historical data could be utilized as the validation prior. However, we adopt the convention of \citet{ibrahim_bayesian_2015} and utilize the most recently completed study for the validation prior. Typically, the most recently completed study will have the larger sample size and will be similar in spirit to the future study. For example, suppose a second Phase III study is being conducted, and we possess historical data from the Phase II study $D_{01}$ and the first Phase III study $D_{02}$. Since the second Phase III study will be similar to the first, it makes sense to generate future data sets based on $D_{02}$, while $D_{01}$ may be utilized at the fitting prior stage to inform the treatment effects.

The validation prior for $\bm{\alpha}$ may be specified as
\begin{equation}
    \pi^{(v)}(\bm{\alpha} | D_{01}, D_{02}, b_{01}, b_{02} )
    \propto
    \left[ \prod_{i = 1}^{n_{01}} f(\bm{x}_{01i} | \bm{\alpha}) \right]^{b_{01}}
    \left[ \prod_{i = 1}^{n_{02}} f(\bm{x}_{02i} | \bm{\alpha}) \right]^{b_{02}}
    \pi_{20}^{(v)}(\bm{\alpha}),
    \label{eq:alphavalidprior2}
\end{equation}
where $\pi_{20}^{(v)}(\bm{\alpha})$ is an initial validation prior. The hyperparameters $b_{01}$ and $b_{02}$ can be chosen in such a way that reflects how similar the study participants in the historical studies are compared to the planned future study. For example, if the patient characteristics in the more recent historical study will be similar to that of the future study and participants of the older historical study will be more different, we may choose $0 \le b_{01} < b_{02} \le 1$. In general, it is advisable to do sensitivity analysis on values of $b_{01}$ and $b_{02}$ to ensure changing their values does not drastically change the resulting POS. Using a mechanism similar to the single data set setting, we may rewrite (\ref{eq:alphavalidprior2}) as
\[
     \pi^{(v)}(\bm{\alpha} | D_{01}, D_{02}, b_{01}, b_{02} ) 
     \propto \prod_{k=1}^2 \prod_{l=1}^L
         \left[ \prod_{i = 1}^{n_{0k}} f(x_{0kil} | x_{0ki1}, \ldots x_{0ki,(l-1)}, \bm{\alpha}_l) \right]^{b_{0k}}
    \pi_{20}^{(v)}(\bm{\alpha}),
    \label{eq:validprior2v2}
\]
where $x_{0kil}$ is the $l^{th}$ covariate from study $k$ for the $i^{th}$ subject and $\alpha_l$ is the parameter vector for the $l^{th}$ covariate given all the covariates before it.

In our development, the validation prior was only discounted for the covariates, while the fitting prior allowed discounting for $\bm{\theta}$. The reasons for doing so is because the validation prior for $\bm{\theta}$ represents belief about the effects to be studied. It does not have any impact on the analysis of a future data set. If the validation prior for $\bm{\theta}$ is discounted, it becomes more likely that unrealistic values for $\bm{\theta}$ are generated, which may underestimate POS. On the other hand, the fitting prior for $\bm{\theta}$ is discounted because the fitting prior is the prior utilized in the analysis of the future data, and the future data should have more weight than the historical data. We allow discounting for the covariate parameters $\bm{\alpha}$ in the validation prior because, in practice, the covariates and characteristics of trial participants from a future trial could be quite different than those from historical data. In practice, discounting the covariate distribution has a marginal effect on POS, as shown in our simulations and data analysis.

\subsection{POS Algorithm}
The algorithm to compute probability of success is as follows:

\begin{enumerate}[leftmargin = 2cm, label=Step \arabic*:]
  \item Specify $n$, $\Omega$, $\gamma$, $q$ (the proportion of subjects randomized to the test drug group), the hyperparameters in the initial validation prior $\pi_2(\bm{\alpha})$, $M$ (the number of samples to draw), and $B$ (the number of future data sets).
  \item Generate $\bm{\theta}$ and $\bm{\alpha}$ from the validation prior $\pi^{(v)}(\bm{\theta}, \bm{\alpha})$.
  \item Using the value of $\bm{\alpha}$ from step 1, generate $\bm{x}_i$ from $f(\bm{x}_i | \bm{\alpha})$ and randomly assign the $i^{th}$ subject into two treatment arms by generating $z_i \sim \text{Bernoulli}(q)$ independently for $i = 1, \ldots, n$ and set $\bm{X}$ to be the resulting $nJ \times (p + J)$ design matrix.
  \item Using the value of $\bm{\theta}$ from Step 1 and the value of $\bm{X}$ from Step 2, generate a future response vector $\bm{y}$ from $\N_{nJ}(\bm{X\beta}, \bm{\Sigma}^{-1} \otimes \bm{I}_n)$ to obtain $D = (\bm{y}, \bm{X})$.
  \item Using $D$ from Step 4, use the posterior distribution using the fitting prior to generate $\left\{ \bm{\theta}^{(m)} = (\bm{\beta}^{(m)}, \bm{\Sigma}^{(m)}) \right\}, m = 1, \ldots, M$
  \item Compute $I^{(b)} = 1\left[ \frac{1}{M} \sum_{m = 1}^M 1\left\{ \bm{\beta}^{(m)} \in \Omega \right\} \ge \gamma \right]$ using $\bm{\theta}^{(m)}$, $m = 1, \ldots, M$, from Step 5
  \item Repeat Step 2 to Step 6 $B$ times to obtain $I^{(1)}, \ldots, I^{(B)}$. Probability of success is given by $\POS = \frac{1}{B} \sum_{b=1}^B I^{(b)}$.
\end{enumerate}


\section{Family-wise error rate control for multiplicity}
\label{sec:fwer}
Computing POS as described will not guarantee strong control of FWER when the set that defines success, $\Omega$, contains unions of events. For example, suppose $J=3$ and $\Omega = \{ \beta_{12} > 0 \cup \beta_{13} > 0 \}$. Suppose further that $\beta_{12}$ and $\beta_{13}$ are restricted to be null. That is, the validation prior (\ref{eq:validprior}) is augmented to
\begin{equation}
    \tilde{\pi}_{1}^{(v)}(\bm{\theta}) = \pi_{1}^{(v)}(\bm{\theta}) 1\{ \bm{\theta} \in \bm{\Theta}_0 \},
    \label{eq:vpriornull}
\end{equation}
where $\Theta_0 = \{ (\bm{\beta}, \bm{\Sigma}) : \beta_{12} = 0 \cap \beta_{13} = 0 \}$. Let $E_j = \{ \beta_{1j} > 0 \}$ for $j = 1, 2, 3$. Due to the asymptotic uniformity of the posterior probability, $P(\beta_{1j} > 0 | D_0, \pi_1^{(v)}) \to 0.5$ as $n_0 \to \infty$ and $\POS(E_j) \to \alpha = 1 - \gamma$ as $n \to \infty$ for $j = 1, 2, 3$, where $\POS(E)$ refers to computing POS where the event $E$ is the success criteria, i.e., computing POS by substituting $E$ for $\Omega$ in (\ref{eq:success}), and $\gamma$ is the threshold achieve success for $\Omega$, e.g., $\gamma = 0.95$. Then $\POS(E_2 \cup E_3) = \POS(E_2) + \POS(E_3) - \POS(E_2 \cap E_3)$. Asymptotically, the first two terms converge to $\alpha$, but the third term converges to a quantity $\alpha^* \le \alpha$ with equality holding if and only if $E_2 \cup E_3 = \E_2 \cap E_3$, i.e., if and only if $\Corr(\beta_{12}, \beta_{13}) \to 1$. Hence, under the validation prior (\ref{eq:vpriornull}), family-wise type I error rate is inflated for union hypotheses.

Consider replacing $\POS(E)$ with the adjustment $\POS^*(E) = \max\{ \alpha, \POS(E) \}$. Then $\POS^*(E_2 \cup E_3) = \POS(E_2) + \POS(E_3) - \POS(E_2 \cap E_3) \to \alpha$ as $n \to \infty$. Thus, replacing $\POS$ with $\POS^*$ asymptotically guarantees precise control of family-wise error rate (FWER). This leads to the following theorem
\begin{theorem}
  Let $\Omega = \cup_{k=1}^K \{ \beta_{1k} > 0 \}$ for some $2 \le K \le J$ with the validation prior given by (\ref{eq:vpriornull}). Then
  \begin{enumerate}
      \item $\POS(\Omega) \ge 1 - \gamma$ with equality holding if and only if $\Corr(\beta_{1j}, \beta_{1k}) \to 1$ as $n \to \infty$ for every $1 \le j \le K, k \ne j$.
      \item $\POS^*(\Omega) \to 1 - \gamma$ as $n \to \infty$.
  \end{enumerate}
\end{theorem}
A proof of Theorem 1 for a general model is provided in Appendix B. Note that we can write $\POS^*(E_2 \cup E_3) = \POS(E_2 \cup E_3) - \max\{ 0, \alpha - \POS(E_2 \cap E_3) \}$, so that we may view $\POS^*()$ as POS with a penalty term, where the penalty is only applied when $\POS(E_2 \cap E_3) < \alpha$, i.e., when there is substantial evidence that $\beta_{12} \le 0$ and $\beta_{13} \le 0$. The adjustment can be applied to compositions of unions and intersections. For example, suppose $\Omega = \{ E_1 \cap \{ E_2 \cup E_3 \} \} \}$, corresponding to meeting the success criterion for the primary endpoint and \emph{at least one} of the secondary endpoints. We can write 
$\Omega = \{ \{ E_1 \cap E_2 \} \cup \{ E_1 \cap E_3 \} \}$, 
and we can then apply the method above to each intersected set. A formal proof for any finite amount of intersections of a finite amount of unions of events is presented in Appendix B.

We note that while correcting the posterior probabilities in such a way has a frequentist motivation, multiplicity corrections are necessary even from a Bayesian perspective (e.g., \citet{scott_bayes_2010}). For example, if $\Corr(\beta_{12}, \beta_{13}) = 0$ and $\Omega = \{ \beta_{12} > 0 \cup \beta_{13} > 0 \}$, $P(\Omega | D_0, \pi_1^{(v)}) \to 0.75$ as $n_0 \to \infty$. Thus, any analysis involving unions of hypotheses will favor the alternative hypothesis, even if the validation prior is specified from the null as that in (\ref{eq:vpriornull}).


\section{Simulation study}
The Comprehensive Post-Acute Stroke Services (COMPASS) \citep{duncan_comprehensive_2017} study was a two-arm, cluster randomized pragmatic trial to evaluate whether a novel care model (i.e., the COMPASS intervention) functional outcomes at 90 days compaerd to usual care (UC). We base our simulations based on the COMPASS study's summary statistics. The complete historical data consisted $n_0 = 981$ participants and $J = 3$ endpoints. We fix all power prior parameters to 1 and assume the following model to generate the historical data:
\[
    \bm{y}_{0i} = \bm{\beta}_1 z_{0i} + \bm{X}_{0i} \bm{\beta}_2 + \bm{\epsilon}_{0i}, \hspace{1cm} i = 1, 2, \ldots, 981,
    \label{eq:gendata}
\]
where 
$\bm{y}_{0i} = (y_{0i1}, y_{0i2}, y_{0i3})'$ 
is a vector of responses corresponding to the $J=3$ endpoints, $\bm{\beta}_1 = (\beta_{11}, \beta_{12}, \beta_{13})'$ 
is the vector of treatment effects, 
$\bm{X}_{0i} = \blkdiag\{ \bm{x}_{01}', \bm{x}_{02}', \bm{x}_{03}' \}$ is a $n_0 J \times p$ block diagonal matrix of controls for endpoints 1, 2, and 3, 
$\bm{\beta}_2 = (\bm{\beta}_{21}', \bm{\beta}_{22}', \bm{\beta}_{23}')'$ is a $p$-dimensional vector of effects for the controls of each endpoint, and 
$\bm{\epsilon}_{0i} \sim N_3(\bm{0}, \bm{\Sigma})$, where 
$\bm{\Sigma}$ is constructed to be positive definite. For the purposes of the simulation, all three endpoints shared the same seven covariates and each had an intercept term ($p = 24)$. The treatment vector was assumed to be 
$\bm{\beta}_1 = (0.0333, 0.1667, 0.5980)'$. 
The standard deviations for endpoints 1, 2, and 3 were, respectively, 0.193, 0.748, and 7.422. We took half the variance from the historical data to keep the computational time manageable due to the different simulation scenarios. Several different correlation matrices were considered. Let $\bm{\rho} = (\rho_{12}, \rho_{13}, \rho_{23})'$ where each $|\rho_{ij}| < 1$ and $\rho_{ij} = \Corr(y_{0ij}, y_{0ik})$ for $i = 1, \ldots, n_0$ and $1 \le j < k \le 3$. The pairwise correlations considered were 
$\bm{\rho}_{\text{HN}} = (-0.3, -0.4, -0.7)'$, 
$\bm{\rho}_{\text{LN}} = (-0.05, -0.1, -0.2)'$,
$\bm{\rho}_{\text{ind}} = (0, 0, 0)'$,
$\bm{\rho}_{\text{LN}} = (0.05, 0.1, 0.2)'$, and
$\bm{\rho}_{\text{HP}} = (0.3, 0.4, 0.7)'$, which correspond to high negative, low negative, independent, low positive, and high positive correlations, respectively.

We compare the family-wise error rate and power of our method against frequentist approaches, where we focus on the Holm procedure described below. The philosophy behind the frequentist approach to family-wise error rate is to maintain a type I error rate no larger than $\alpha = 1- \gamma$ whenever \emph{at least one} of the null hypotheses is true.

We explore several definitions of success. We define the following sets: 
$\Omega_{j|k} = \{ \bm{\theta} : \beta_{1j} > 0 \cup \beta_{2j} > 0 \}$ for $j \ne k$ and 
$\Omega_{j(k|l)} = \{ \bm{\theta} : \beta_{1j} > 0 \cap \{ \beta_{1k} > 0 \cup \beta_{1l} > 0 \} \}$ for $j = 1, 2, 3, j \ne k < l$. The first set, $\Omega_{j|k}$, corresponds to the case where there are two endpoints defined to be primary, and the trial is declared successful whenever at least one of them is statistically significant. The set $\Omega_{j(k|l)}$ corresponds to achieving the success in the primary endpoint $y_{.j}$ and at least one of two secondary endpoints, $y_{.k}$ and $y_{.l}$. Success for $\Omega_{j|k}$ is declared utilizing the Holm procedure whenever $\min\{p_j, p_k\} < (1 - \gamma)/2$ or $\max\{ p_j, p_k \} < 1 - \gamma$, where $p_j$ and $p_k$ respectively denote the $p$-values for the future data utilizing a frequentist marginal linear regression model. Similarly, utilizing the intersection union test (IUT) principle, success using the Holm procedure for testing $\Omega_{i(j|k)}$ is established when
$p_i < \alpha$ and either $\min\{ p_j, p_k \} < \alpha/2$ or $\max\{ p_j, p_k \} < \alpha$. Bayesian analysis of the SUR model was conducted via the surbayes package (2020) \nocite{noauthor_surbayes_nodate} in R. 

The validation priors utilized to obtain type I error and Bayesian conditional expected power (BCEP) are, respectively,
\[
    \pi_{\overline{\Omega}}^{(v)}(\bm{\theta}) \propto       
      \pi_{1}^{(v)}(\bm{\theta})
      1\left\{ \bm{\theta} \in \overline{\Omega} \right\}
      \label{eq:nullprior}
\]
and
\begin{equation}
   \pi_{\Omega}^{(v)}(\bm{\theta}) \propto       
      \pi_{1}^{(v)}(\bm{\theta})
      1\left\{ \bm{\theta} \in \Omega \right\},
      \label{eq:altprior}
\end{equation}
where $\Omega$ is the set that defines success and $\overline{\Omega}$ is the set defined by the boundary of the null hypothesis, where type I error rate is at its maximum, i.e., $\overline{\Omega}_{j|k} = \{ \bm{\theta} : \beta_{1j} = 0 \cap \beta_{1k} = 0 \}$ and $\overline{\Omega}_{i(j|k)} = \overline{\Omega}_{j|k}$. Computing POS utilizing the validation prior (\ref{eq:altprior}) is referred to as Bayesian conditional expected power because it is conditional on the alternative hypothesis being true.

The first column of Figure~\ref{fig:sim1or2} depicts FWER for the adjusted and unadjusted fully Bayesian methods and the hybrid Holm approach where the success criteria is $\Omega_{1|2}$, i.e., the trial is declared successful when the treatment effects for at least one of endpoints 1 and 2 are significant. The unadjusted version of POS has inflated FWER, but the hybrid Holm and adjusted fully Bayesian POS have FWER control, albeit the Holm approach is slightly more conservative. The adjusted and unadjusted Bayesian approaches are virtually identical in terms of power. This is not surprising because the two methods differ only when there is substantial evidence that the treatment effects for endpoints 1 and 2 are both null. However, the validation prior (\ref{eq:altprior}) precludes this from happening as at least one of the treatment effects must be positive.

\begin{figure}
    \centering
    \includegraphics[width = \textwidth, height=\textheight, keepaspectratio]{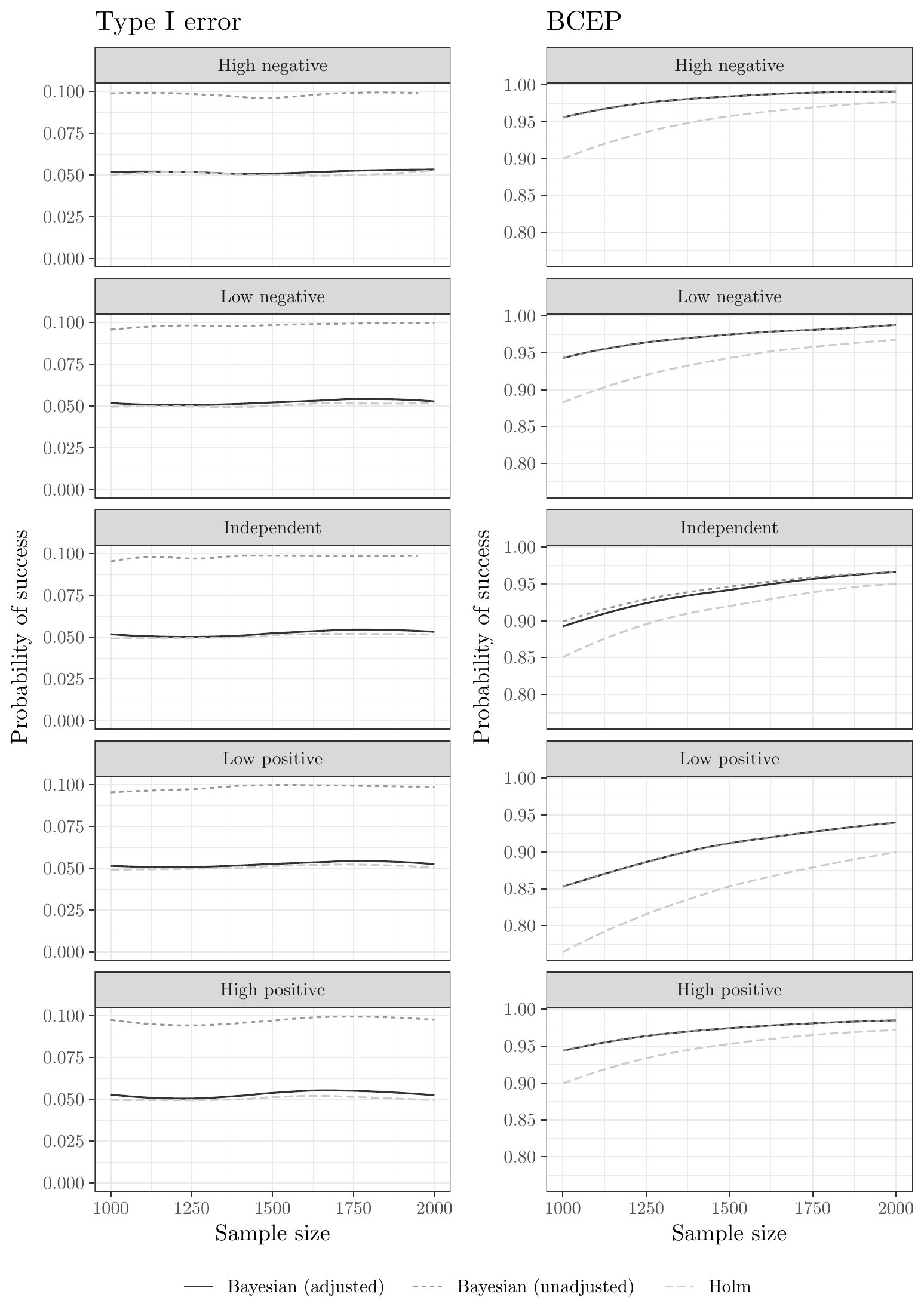}
    \caption{Type I error rate and Bayesian conditional expected power (BCEP) of adjusted and unadjusted fully Bayesian POS with a hybrid Holm procedure where $\Omega_{1|2}$ defines the set that determines success and when $\gamma = 0.95$.}
    \label{fig:sim1or2}
\end{figure}

Suppose now that success is defined when efficacy is established in the primary endpoint and at least one of the two secondary endpoints. Figure~\ref{fig:simcompos} reports type I error rate in the first column and BCEP in the second column for the set $\Omega_{1(2|3)}$, that is, when the primary endpoint is outcome 1 and the secondary endpoints are outcomes 2 and 3. Both the IUT-Holm procedure and the adjusted Bayesian POS control type I error rate, while power is virtually identical between the two methods across all correlations except low positive, where the Holm procedure has marginally lower BCEP. Similar to the previous case, unadjusted Bayesian POS does not control FWER, but adjusted and unadjusted Bayesian POS have indistinguishable FWER. 

\begin{figure}
    \includegraphics[width = \textwidth, height=\textheight, keepaspectratio]{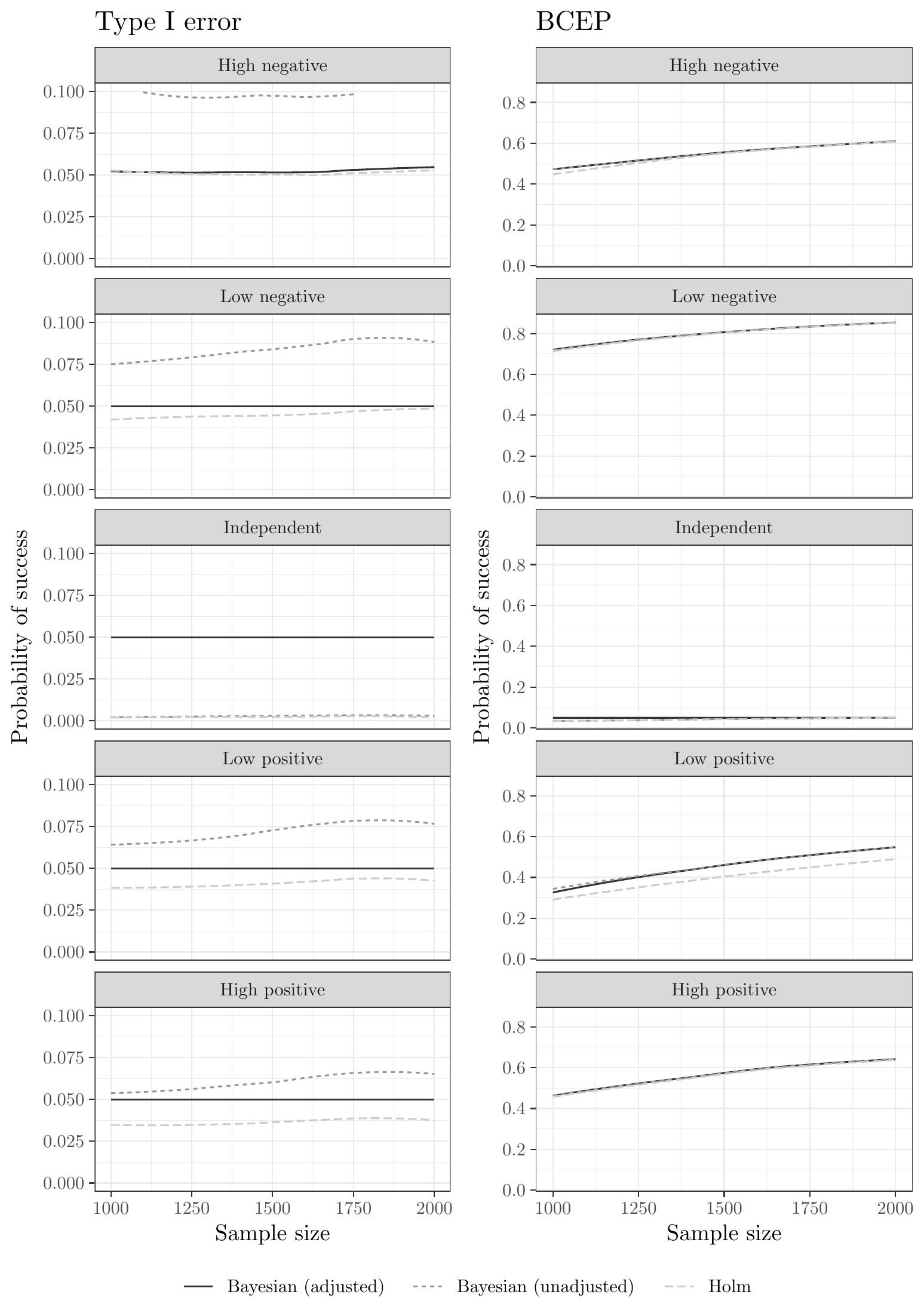}
    \caption{Type I error rate and Bayesian conditional expected power (BCEP) of adjusted and unadjusted fully Bayesian POS with a hybrid Holm procedure where $\Omega_{1(2|3)}$ defines the set that determines success and when $\gamma = 0.95$.}
    \label{fig:simcompos}
\end{figure}

The simulations suggest that our method is uniformly more powerful than the Holm procedure. While we do not present a formal proof, a heuristic argument is that our method asymptotically provides exact FWER of $1-\gamma$, while the Holm procedure only guarantees FWER no more than $1 - \gamma$. That is, the Holm procedure may be conservative.

\section{Data Application}
In Phase 1 of the COMPASS study, a diverse collection of 40 North Carolina, USA hospitals were randomized in a 1:1 allocation scheme to administer their patients the COMPASS intervention or SOC. The hospitals that continued to Phase 2 of the study attempted to continue providing the intervention with minimal resources provided by the study team. Phase 1 hospitals that were randomized to deliver SOC to their patients were transitioned to provide the intervention.

The primary outcome was stroke impact scale (SIS-16), which is a measurement between 0 and 100 indicating physical function. There were several secondary endpoints in the study. The continuous endpoints we utilize are self-reported general health and physical health. Self-reported general health was measured by an approximately continuous 5-point Likert scale. Physical health was measured using the PROMIS Global Health Scale \citep{hays_development_2009}, which is a continuous measure of overall physical health, mental health, social health, pain, fatigue, and overall perceived quality of life, and was collected in only Phase 2 of the study. As a result, we only consider Phase 2 data in our analysis. For the purposes of demonstration of our method, we partition the full data between a vanguard site (considered the older historical data $D_{01}$) and the remaining data (considered the newer historical data $D_{02}$).

We assume the following model for each outcome:
\[
    y_{ij} = \beta_0 + \bm{\beta}_1 z_i + \bm{x}_{i}' \bm{\beta}_{2j} + \epsilon_{ij},
\]
where $\bm{\epsilon}_i = (\epsilon_{i1}, \epsilon_{i2}, \epsilon_{i3})'$ are independent and identically distributed (i.i.d.) $\N_3(\bm{0}, \bm{\Sigma})$ where $\bm{\Sigma}$ is a $3 \times 3$ covariance matrix assumed to be positive definite, $z_i = 1$ if subject $i$ received the eCare plan intervention and 0 otherwise, $\bm{x}_{i}$ is a vector of control variables, which include stroke history, TIA history, linear and quadratic terms of age, race (white or non-white), severity of stroke, whether the patient had insurance (to control for socioeconomic status), and a binary variable indicating whether the patient was hospitalized due to a stroke or a transient ischemic attack (TIA). In order to reduce the skew and asymmetry of the distribution of the outcome, we obtained the log transform of the stroke impact scale. Bayesian analysis of the SUR model was conducted via the surbayes package (2020) \nocite{noauthor_surbayes_nodate} in R. 

Define $\Omega_{1} = \{ \beta_{11} > 0 \}$, $\Omega_{123} = \{ \beta_{11} > 0 \} \cap \{ \beta_{12} > 0 \} \cap \{ \beta_{13} < 0 \}$, and $\Omega_{1(2|3)} = \{ \beta_{11} > 0 \} \cap \left(\{ \beta_{12} > 0 \} \cup \{ \beta_{13} < 0 \} \right)$. These sets correspond to achieving success in the primary endpoint, in all three endpoints, and in the primary endpoint and at least one of the two secondary endpoints, respectively. For sensitivity of power prior parameter specification, we checked all possible combinations of power prior parameters and future sample sizes for the following: 
$n      \in \{3000, 3100, \ldots, 4000 \}$, 
$b_{01} \in \{0, 0.25, 0.5, 0.75, 1 \}$, 
$b_{02} \in \{0.25, 0.5, 0.75, 1 \}$, 
$a_0    \in \{ 0, 0.25, 0.5, 0.75, 1 \}$,
resulting in $11 \times 5 \times 4 \times 5 = 1,100$ scenarios. We set the number of future data sets $B = 10,000$ and posterior sample size $M = 10,000$. 

Table~\ref{tab:pos1} shows adjusted POS, $\POS^*$ (henceforth referred to as simply POS), for $n = 4,000$ across the differing power prior parameters. POS computations were not sensitive to the choice of the power prior parameters. Across all simulation scenarios, POS for the primary endpoint is over 80 percent. However, is only just over 40 percent for achieving all three endpoints, and it is above 70 percent for achieving success in the primary outcome and at least one of the two secondary outcomes. This suggests that a future study with high probability of success to meet all 3 endpoints would require an extremely large sample size, but only a moderately larger sample size would be required to achieve 80\% probability of success in meeting the primary endpoint and at least one secondary endpoint.

\begin{table}[ht]
\let\center\empty
\let\endcenter\relax
\centering
\resizebox{\textwidth}{!}{\input{Tables/pos}}
\caption{Adjusted probability of success across different criteria of success for the COMPASS data set for $n = 4,000$. $\Omega_1$, $\Omega_{123}$, and $\Omega_{1|23}$ respectively refer to achieving success in the first endpoint alone, for all three endpoints, and for the first endpoint and at least one of the two secondary endpoints, respectively. The quantity $a_0$ gives the power prior parameter for the fitting prior, and $b_{01}$ and $b_{02}$ respectively refer to the power prior parameters for the covariates in the validation prior for the older historical and newer historical data set. Zero parameters indicate the corresponding historical data set was not used in the computations.}
\label{tab:pos1}
\end{table}

Since our choice of power prior parameters is not sensitive to the resulting POS, we proceed with the discussion of the results only considering power prior parameters of 0 and 1. Figure~\ref{fig:datapos} shows the POS curves for the COMPASS data. Across all three definitions of success, POS is several percentage points higher when using the older historical data set in the fitting prior ($a_0 = 1$) than when using an uninformative validation prior ($a_0 = 0)$. The covariate validation prior parameters $b_{01}$ and $b_{02}$ have only a marginal impact on POS. The sharp, upward sloping curve for achieving success in all three endpoints, $\Omega_{123}$, indicates that the return for increasing sample size for this endpoint could be high. However, the sample size to achieve any reasonable level of success would likely need to be prohibitively large.

\begin{figure}
    \centering
    \includegraphics[width = \textwidth, height=\textheight, keepaspectratio]{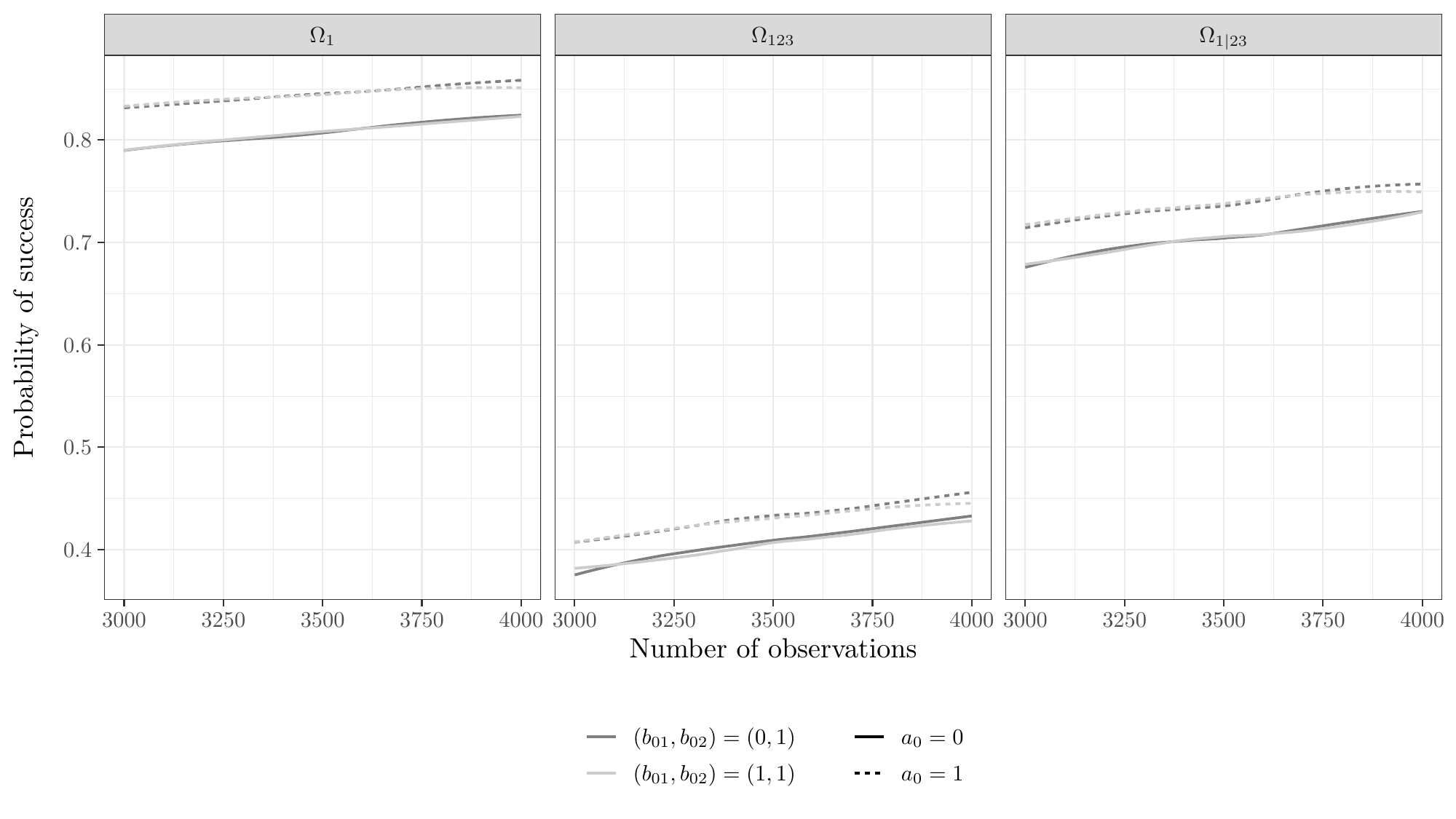}
    \caption{Adjusted probability of success curve for the COMPASS data. The quantities $b_{01}$ and $b_{02}$ denote the power prior hyperparameters for the older and newer historical data sets, respectively, when fitting the validation prior for the covariates. The hyperparameter $a_0$ is the power prior parameter for the older historical data set in the fitting prior. The solid lines indicate the hyperparameter $a_0 = 0$ and the dashed lines indicate $a_0 = 1$. Power prior parameters of 0 indicate that the historical data set was unused.}
    \label{fig:datapos}
\end{figure}

\section{Discussion}
\label{sec:disc}
There are several advantages to our approach compared to previous POS approaches (e.g., \citet{ibrahim_bayesian_2015} and \citet{chuang-stein_sample_2006}). First, the SUR model allows clinicians to have different covariates across endpoints. For example, a practitioner may wish to control for the baseline level of the endpoint of interest, but not the baseline level of other endpoints. Second, by using sets to define success, our method makes no distinction between testing simple or multiple hypotheses. This unification of testing lends itself to much simpler interpretation than adjusted $p$-values, i.e., POS always refers to expected power across a range of plausible effect sizes.

For practitioners, it is worth asking when an analysis consiting of marginal regressions should be conducted as opposed to a joint analysis. \citet{berry_bayesian_1999} argue that it is more realistic to view parameters as dependent in multiple testing situations. In particular, knowledge of the value of one parameter should influence the value of other parameters when study data are correlated, and they argue it seldom makes sense to analyze the data independently. Moreover, we have established a Bayesian approach that asymptotically guarantees strict control of FWER regardless of the posterior correlation between the parameters. It is important to differentiate our method from a Bayesian posterior model probability, although the two are related. While \citet{scott_bayes_2010} argue that the only way to control for multiplicity in a Bayesian perspective is via prior model probabilities, our approach does not rely on prior model probabilities, which are typically subjective; nor does it require all priors in the model to be proper.

The method may also be applied when historical data sets are small. One may keep only the $q^{th}$ highest posterior density (HPD) region of the treatment effects when sampling from the validation prior. This may be done by either evaluating the posterior likelihood or utilizing kernel density estimation (KDE) to obtain a marginal joint posterior for just the treatment effects. A description of these adjustments is provided in Appendix C. We also apply these validation prior adjustment methods based on a simulated historical data set. A comparison of their effects is presented in Figure \ref{fig:ivacaftor}. Figure~\ref{fig:ivacaftor} suggests that adjusting the validation prior using HPD regions tends to increase POS when there is high power to detect an effect. When the power is low, adjusting the validation prior has only a marginal effect on POS. Similarly, the power prior parameters of the covariate distribution has a marginal effect on overall POS.

We close by emphasizing that the proposed method is useful regardless of whether a Bayesian or classical frequentist analysis will ultimately be used to analyze the future study. If several endpoints are being measured, the measurements will be correlated. Thus, it would be unintuitive to compute POS at each endpoint individually--one must sample from the validation prior at the joint level to account for correlations in the response variables.



\section*{Acknowledgements}
This research was supported in part by a grant from the National Institute of Environmental Health Sciences (T32ES007018).




\bibliography{references}

\section*{Appendix}


\subsection*{Appendix A: Full conditional distributions for the power prior for the SUR model}
Suppose that we have historical data $D_0 = \{(\bm{y}_{0j}, \bm{X}_{0j}), j = 1, \ldots, J \}$. Let $\bm{y}_0 = (\bm{y}_{01}, \ldots, \bm{y}_{0J})'$ and let $\bm{X}_0 = \blkdiag\{ \bm{X}_{01}, \ldots, \bm{X}_{0J} \}$ The fitting prior is given by
\begin{align}
    \pi^{(f)}(\bm{\beta}, \bm{\Sigma} | \pi_{10}^{(f)}, a_0) &\propto 
      \left[ 
        \lvert \bm{\Sigma} \rvert^{n_0 / 2} 
        \exp\left\{
          -\frac{1}{2} (\bm{y}_0 - \bm{X}_0 \bm{\beta})' 
            (\bm{\Sigma}^{-1} \otimes {\bm{I}}_{n_0})
            (\bm{y}_0 - \bm{X}_0 \bm{\beta})
        \right\}
      \right]^{a_0}
      \pi_{10}^{(f)}(\bm{\beta}, \bm{\Sigma}).
\end{align}
\subsubsection*{The conditional posterior of $\bm{\beta}$}

Let $\bm{QQ}' = a_0 \bm{\Sigma}^{-1} \otimes \bm{I}_{n_0}$ represent the Cholesky decomposition and let $\bm{y}_0^* = \bm{Q}'\bm{y}_0$ and $\bm{X}_0^* = \bm{Q}' \bm{X}_0$. If we specify $\pi^{(f)}_{10}(\bm{\beta}, \bm{\Sigma})$ as (\ref{eq:initvalidprior}), then
\begin{align}
    \pi^{(f)}(\bm{\beta} | \bm{\Sigma}, a_0) 
    &\propto \N_p \left( \bm{\beta} | \hat{\bm{\beta}}_0, \hat{\bm{\Sigma}}_{\bm{\beta}_0} \right),
     \label{eq:condbetafitprior2}
\end{align}
where $\N_p(\cdot | \bm{a}, \bm{B})$ is the $p$-dimensional multivariate normal distribution with mean $\bm{a}$ and covariance matrix $\bm{B}$, and \\
$\hat{\bm{\Sigma}}_{\bm{\beta}_0} = a_0^{-1} \left[ \bm{X}_0' \left( \bm{\Sigma}^{-1} \otimes \bm{I}_{n_0} \right) \bm{X}_0 \right]^{-1}$
and
$\hat{\bm{\beta}}_0 = \left[ \bm{X}_0' \left( \bm{\Sigma}^{-1} \otimes \bm{I}_{n_0} \right) \bm{X}_0 \right]^{-1} \bm{X}_0' \left( \bm{\Sigma}^{-1} \otimes \bm{I}_{n_0}  \right) \bm{y}_0$.

Using a similar argument, we can then write the likelihood of the new data as
\begin{equation}
    L(\bm{\beta} |\bm{y}, \bm{X}, \bm{\Sigma}) \propto 
      \exp\left\{ -\frac{1}{2}
        (\bm{y}^* - \bm{X}^* \bm{\beta})'
        (\bm{y}^* - \bm{X}^* \bm{\beta})
        \label{eq:condlikebeta}
    \right\}
\end{equation}
where $\bm{y}^* = \bm{R}'\bm{y}$,
$\bm{X}^* = \bm{R}'\bm{X}$,
and $\bm{R}$ is the lower triangular matrix from the Cholesky decomposition of $\bm{\Sigma} \otimes \bm{I}_{n}$. The likelihood (\ref{eq:condlikebeta}) is just a normal linear regression likelihood with unit variance. The posterior of the normal linear regression model with a normal prior has been well studied. Thus, the conditional posterior distribution of $\bm{\beta}$ is given by
\[
    \pi(\bm{\beta} | \bm{y}, \bm{\Sigma}, a_0, \pi^{(f)} ) =
      \N_p \left(\bm{\beta} | \tilde{\bm{\beta}}, \tilde{\bm{\Sigma}}_{\beta} \right),
\]
where
${\tilde{\bm{\Sigma}}}_{\bm{\beta}} = %
\left[ %
\bm{X}' \left( \bm{\Sigma}^{-1} \otimes \bm{I}_n \right) \bm{X} + %
a_0^{-1} \bm{X}'_0 \left( \bm{\Sigma}^{-1} \otimes \bm{I}_{n_0} \right) \bm{X}_0 %
\right]^{-1}$
and \\
$\tilde{\bm{\beta}} = \tilde{\bm{\Sigma}}_{\bm{\beta}} \left[ \bm{X}' \left( \bm{\Sigma}^{-1} \otimes \bm{I}_n \right) \bm{y} + \bm{X}_0'(\bm{\Sigma}^{-1} \otimes \bm{I}_{n_0}) \bm{y}_0 \right]$.

\subsubsection*{The conditional posterior of $\bm{\Sigma}$}
Suppose that the initial fitting prior is specified as (\ref{eq:initvalidprior}). Then
\begin{align}
    \pi(\bm{\Sigma}^{-1} &| \bm{\beta}, D, D_0, a_0)  \notag \\
    &= \lvert \bm{\Sigma}^{-1} \rvert^{n/2} \exp\left\{ -\frac{1}{2} \tr(\bm{R} \bm{\Sigma}^{-1}) \right\} \times
      \lvert \bm{\Sigma}^{-1} \rvert^{a_0 n_0/2} \exp\left\{ -\frac{1}{2} \tr(a_0 \bm{R}_0 \bm{\Sigma}^{-1}) \right\} \times
      \lvert \bm{\Sigma}^{-1} \rvert ^{-(J + 1) / 2} \\
      &= \lvert \bm{\Sigma}^{-1} \rvert ^{(n + a_0 n_0 - J - 1) / 2 }
         \exp\left\{ -\frac{1}{2} \tr\left( \left[\bm{R} + a_0 \bm{R}_0 \right] \bm{\Sigma}^{-1} \right) \right\} \\
      &\propto \W_{J}(\bm{\Sigma}^{-1} | \tilde{\bm{V}}, \tilde{\nu}), \label{eq:condpostsigma}
\end{align}
where $W_{J}(\cdot | \bm{C}, d)$ is the probability density function for a $J$-dimensional Wishart distribution with scale matrix $C$ and $d$ degrees of freedom, 
$\tilde{\bm{V}} = \tilde{\bm{V}}(\bm{\beta}) = \left[ \bm{R}(\bm{\beta}) + a_0 \bm{R}_0(\bm{\beta}) \right]^{-1}$,
$\tilde{\nu} = n + a_0 n_0$,
and $\bm{R}$ and $\bm{R}_0$ are the matrices of dot-product residuals for the future data and historical data, respectively, as described in (\ref{eq:sur2}).

\subsection*{Appendix B: Multiplicity adjustment for union hypotheses}

\begin{theorem}
\noindent Let $\bm{\nu} = (\bm{\theta}', \bm{\eta}')'$ be a $p$-dimensional vector of parameters where we are interested in subsets of the $J$-dimensional vector $\bm{\theta}$ and $\bm{\eta}$ is a $p - J$ vector of nuisance parameters. Suppose we obtain data of size $n$, denoted by $D_n$, from the likelihood function $L(\bm{\nu} | D_n)$ Let the posterior density of $\bm{\nu}$ be given by
\begin{align}
    \pi^{(v)}(\bm{\nu} | D_n) \propto L(\bm{\nu} | D_n) \pi^{(v)}_{10}(\bm{\nu}) 1\{ \bm{\theta} \in \Theta_0 \},
    \label{eq:vpriorproof}
\end{align}
where $\pi^{(v)}_{10}(\bm{\nu})$ is a prior on $\alpha$ constrained to make make the posterior $\pi(\bm{\nu} | D_n)$ proper and $\Theta_0 = \{ \bm{\theta} : \theta_{j_k} = \mu_{j_k}, k = 1, \ldots, K \}$ for some $K \le J$. Henceforth, we will assume that the first $K$ elements of $\bm{\theta}$ are the ones for which we wish to conduct inference. Further, let $\pi^{(f)}(\bm{\nu})$ denote the fitting prior. Define POS based on $K$ events as
\begin{align}
    \hbox{POS}(\cup_{k=1}^K E_k) = \int 
    P\left[ P(\cup_{k=1}^K E_k | D, \pi^{(f)}) \ge \gamma \right] \pi^{(v)}(\bm{\nu} | D_n)
    d \bm{\nu}
    ,
\end{align}
where $E_k = \{ \theta_k > \mu_k \}$ for $k = 1, \ldots, K$. Then $\hbox{POS} \to 1 - \gamma^* \ge 1 - \gamma$, with strict inequality holding whenever $\rho_{kl} = \lim_{n \to \infty} \Corr(\theta_k, \theta_l | D_n) \not \to 1$ for every $k = 1, \ldots, K$, $k \ne l$, as $n \to \infty$.
\end{theorem}

\begin{proof}
  By the Bernstein-von Mises theorem, $\bm{\nu} | D_n \to \N_p(\hat{\bm{\nu}}, I^{-1}(\hat{\bm{\nu}})) $ as $n \to \infty$, where $\hat{\bm{\nu}} = (\hat{\bm{\theta}}', \hat{\bm{\eta}}')'$ is the posterior mode of $\bm{\nu}$ and $I(\bm{\nu}) = -\lim_{n \to \infty} n^{-1}\dfrac{\partial^2 \log L(\bm{\alpha} | D_n)}{\partial \bm{\nu} \partial \bm{\nu}'}$ is the asymptotic Hessian matrix. Thus, $(\theta_1, \ldots \theta_K)' \to \N_K(\hat{\bm{\theta}}_K, \bm{\Sigma}_K)$ as $n \to \infty$, where $\hat{\bm{\theta}}_K$ is the first $K$ elements of $\bm{\hat}{\bm{\theta}}$ and $\bm{\Sigma}_K$ is the upper $K \times K$ matrix of $\bm{I}^{-1}(\bm{\nu})$. Utilizing (\ref{eq:vpriorproof}), $\hat{\bm{\theta}}_K = \bm{\mu} = (\mu_1, \ldots, \mu_K)'$ and hence, $P(E_k | D_n) \to 0.5$ as $n \to \infty$ for every $k = 1, \ldots, K$. It follows that
  \begin{align}
      \POS(E_k) := \int P[ P(E_k | D, \pi^{(f)}) \ge \gamma] \pi^{(v)}(\bm{\nu} | D_n) d \bm{\nu} \to 1 - \gamma := \alpha
      \label{eq:marginalconv}
  \end{align}
  for every $k = 1, \ldots, K$.
  
 \noindent \underline{Base case}: $K = 2$. Note that
  \begin{align}
      \POS(E_1 \cup E_2) &= \sum_{k=1}^2 \POS(E_k) - \POS(E_1 \cap E_2).
      \label{eq:twoway}
  \end{align}
  By the (\ref{eq:marginalconv}), the quantity in the sum of (\ref{eq:twoway}) converges to $2\alpha$. Since $(E_1 \cap E_2) \subset E_k$ for $k = 1,2$, $\POS_{1,2} \to \alpha^* \le \alpha$. Due to the definitions of $E_1$ and $E_2$, $\alpha^* = \alpha$ if and only if $\rho_{12} = 1$. Hence, $\POS \ge 1 - \gamma$ for $K = 2$ with strict inequality holding whenever $\rho_{12} \ne 1$.
  
  \noindent \underline{Induction step:} Suppose $\POS(\cup_{k=1}^K E_k) \ge 1 - \gamma = \alpha$ for some $k = 2, \ldots, K \le J$ with strict inequality holding whenever $\rho_{jk} \ne 1$ for every $j = 1, \ldots, K$ and $k \ne j$. Then
  \begin{align}
      \POS(\cup_{k=1}^{K+1} E_k) &= \int P\left[ P(\cup_{k=1}^{K+1} E_k | D, \pi^{(f)}) \ge \gamma \right] \pi^{(v)}(\bm{\nu} | D_n)
    d \bm{\nu} \\
    &=  \int P\left[ P(\{\cup_{k=1}^{K} E_k\} \cup E_{K+1} | D, \pi^{(f)}) \ge \gamma \right] \pi^{(v)}(\bm{\nu} | D_n)
    d \bm{\nu} \\
    &= \POS(\cup_{k=1}^K E_k) + \POS(E_{K+1}) - \POS(E_{K+1} \cap \cup_{k=1}^K E_k).
    \label{eq:inductionpos}
  \end{align}
 By the induction hypothesis, $\POS(\cup_{k=1}^K E_k) \ge \alpha$ with strict inequality holding whenever $\rho_{kl} \ne 1$ for any $k \ne l, k = 1, \ldots, K$. By (\ref{eq:marginalconv}), $\POS(E_{K+1}) \to \alpha$. Thus, it is clear from (\ref{eq:inductionpos}) that $\POS(K+1) \to \alpha^* \ge \alpha$ for some $\alpha^*$ since $\POS_{K+1} \ge \POS_{(1 | ,\ldots, | K), K+1}$. Now, suppose $\rho_{kl} = 1$ for $k \ne l, j = 1, 2, \ldots, K$ so that $\POS(K) = \alpha$. Note that $P(\cup_{k=1}^K E_k \cap E_{K+1}) \le P(E_{K+1})$ with strict inequality holding if and only if $E_{K+1} = E_{K+1} \cap (\cup_{k=1}^K E_k)$. Asymptotically, this is equivalent to $\rho_{k,K+1} \ne 1$ for some $k = 1, \ldots, K$.
\end{proof}

\begin{theorem}
  Begin with the same setup as in the previous theorem. Define the function $\POS_{\text{adj}}$ by
  \begin{align}
      \POS_{\text{adj}}(\cup_{k=1}^K E_k) = \sum_{k=1}^K \left( (-1)^{k-1} \sum_{\substack{I \subset \{ 1, \ldots, K \} \\ |I| = k}} \POS^*(E_I) \right),
  \end{align}
  where the equality holds by the inclusion / exclusion principle, $E_I = \cap_{i \in I} E_i$, and $\POS^*(E_I) := \max\{ 1 - \gamma, \POS(E_I) \}$. Then $\POS_{adj}(\cup_{k=1}^K E_k) \le 1 - \gamma = \alpha$.
  
  \begin{proof}
    Note that under the specified validation prior, $\POS^*(E_I) = \max\{ 1 - \gamma, \POS(E_I) \} \to 1 - \gamma$ since, as shown in the previous theorem, $\POS(E_I) \le 1 - \gamma$. Hence,
    \begin{align}
        \POS_{adj}(\cup_{k=1}^K E_k) 
          &= \sum_{k=1}^K \left( (-1)^{k-1} \sum_{\substack{I \subset \{ 1, \ldots, k \} \\ |I| = k}} \POS^*(E_I) \right) \\
        &= \sum_{k=1}^K \left( (-1)^{k-1} \sum_{\substack{I \subset \{ 1, \ldots, k \} \\ |I| = k}} (1 - \gamma) \right) \\
        &= (1 - \gamma) \sum_{k=1}^K \left( (-1)^{k-1} \sum_{\substack{I \subset \{ 1, \ldots, k \} \\ |I| = k}} 1 \right) \\
        &= (1 - \gamma) \sum_{k=1}^K  (-1)^{k-1} \binom{K}{k} \\
        &= (1 - \gamma) (1) \\
        &= 1 - \gamma
    \end{align}
  \end{proof}
\end{theorem}

\begin{corollary}
  Let $\{ F_m \}_{m=1}^M$ be a family of sets of the form
  \begin{align}
      F_m = \cup_{i_m = 1}^{I_m} E_{i_m},
  \end{align}
  where $E_{i}$ is defined as in the theorems above. Suppose we are interested in computing POS based on $\Omega = \cap_{m=1}^M F_m$. Then $\POS_{adj}(\Omega) \le 1 - \gamma$.
  
  \begin{proof}
    Note that we can write
    \begin{align}
        \Omega &= \cap_{m=1}^M F_m \\
               &= \cap_{m=1}^M \cup_{i_m = 1}^{I_m} E_{i_m} \\
               &= \cup_{i_1=1}^{I_1} \cdots \cup_{i_m = 1}^{I_m}
                  \cap_{m=1}^M E_{i_m}
    \end{align}
    The result is immediate by applying Theorem A.2 and the inclusion / exclusion principle.
  \end{proof}
\end{corollary}

\textbf{Appendix C: A note on small historical data sets}\\
\\
While using uninformative priors may be suitable for studies such as the COMPASS study, where the historical data set sample size is large, this is relatively uncommon in clinical trials. Typically, Phase II trials will consist of no more than a few hundred participants. Oftentimes, these trials will consist of only several dozens of participants. The validation prior (\ref{eq:validprior}) would be unsuitable for small historical data sets since the posterior variance of the treatment effect would be quite large. 

One potential solution is to replace (\ref{eq:initvalidprior}) with an informative initial validation prior. However, doing so requires prior knowledge on the treatment effect, which is unrealistic in Phase II trials particularly for novel drugs. Moreover, if the chosen prior is informative in $\bm{\Sigma}$, the direct Monte Carlo method cannot be used. Another solution would be to only consider Bayesian conditional expected power by constraining the validation prior to sample only values that indicate efficacy. However, the large variance for small sample sizes might result in large, unrealistic values for the treatment effect, and one may grossly overestimate or underestimate POS.

Instead of the previous methods, we propose to subset the $q^{th}$ highest posterior density (HPD) region of the posterior samples. The primary advantage of this method is it is trivial to conduct: simply take a large number of samples of 
$(\bm{\beta}, \bm{\Sigma})$, compute the logarithm of the posterior for each sample using (\ref{eq:postfit1}), and keep the $q$\% of samples that yield the largest values for this log posterior. However, if the number of parameters is large (which is typically the case for models with more than 2 endpoints), one might have to choose a minuscule $q$ in order to constrain realistic values for the posterior samples of the treatment effects. Doing this will restrict samples from the validation prior to be more reasonable than what would be obtained without adjustment.

A second solution is to use kernel density estimation (KDE) to estimate the marginal joint posterior density of the treatment effects, taking the largest $q$\% of the samples in terms of the density. The primary advantage of this method over the log posterior method HPD method mentioned above is that it allows one to focus on the treatment effects. However, the method is much more computationally intensive, and can take quite a long time if the number of samples and/or the number of endpoints is large. 

We apply both of these solutions to a simulated historical data set based on summary statistics from a Phase II Cystic Fibrosis trial analyzing the safety of Ivacaftor, formerly known as VX-770 \citep{accurso_effect_2010}, available on clinicaltrials.gov. The phase II study consisted of 16 patients (8 participants received placebo, 8 received treatment). The Phase III study enrolled 69 patients \citep{moss_efficacy_2015}. The primary endpoint for the Phase III study was absolute change from baseline in percent predicted forced expiratory volume (FEV1) in 1 second through week 24. Two of the continuous secondary endpoints measured in Phase II and Phase III were change from baseline in sweat chloride and change from baseline in the Cystic Fibrosis questionnaire-Revised (CFQ-R) score. The values for the treatment effects were, respectively, $\bm{\beta}_1 = (6.4, -49.1, 3.5)'$, and the standard deviations for each endpoint were $(\sigma_1, \sigma_2, \sigma_3) = (5.12, 12.26, 7.05)$. Covariates, including linear and squared terms of age, weight, body mass index (BMI), and gender, were included in the models, but each their effects on treatment were assumed to be 0 (i.e., $\bm{\beta}_2 = \bm{0}$). Bayesian analysis of the SUR Model was conducted via the surbayes package (2020) in R \nocite{noauthor_surbayes_nodate}.

The results for some of success criteria are depicted in Figure~\ref{fig:ivacaftor}. Using a HPD region adjustment typically increases the likelihood of success. However, as CFQ-R is underpowered, the effect of adjusting the validation prior does not always result in increasing POS. The different hyperparameters for the power prior for the covariates resulted in only marginal differences in POS.

\renewcommand\thefigure{A\arabic{figure}} 
\setcounter{figure}{0}   
\begin{figure}[H]
    \centering
    \includegraphics[width = \textwidth, height=\textheight, keepaspectratio]{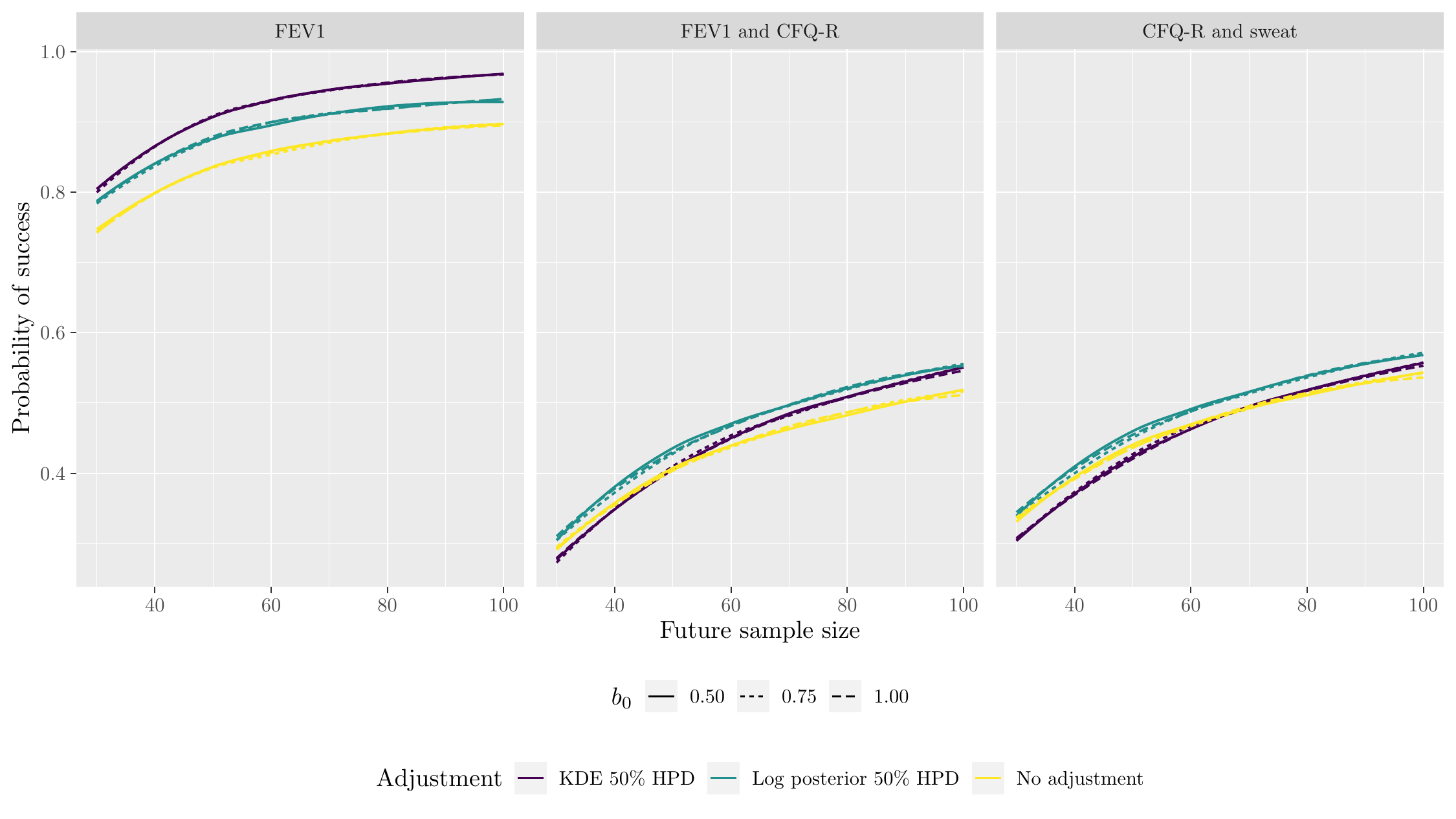}
    \caption{Effects on adjusting the validation prior on POS for the simulated Ivacaftor data when $\gamma = 0.95$. Three different success definitions were used: the primary endpoint (FEV1), the primary endpoint and CFQ-R, and the two secondary endpoints jointly. The quantity $b_0$ refers to the power prior parameter used for the covariate distribution.}
    \label{fig:ivacaftor}
\end{figure}

\end{document}

%% file: Tables/pos.tex
\begin{tabular}{lccccccccccccccc}
\toprule
 & \multicolumn{15}{c}{$a_0$} \\ \cmidrule(lr){2-16}
 & \multicolumn{3}{c}{0} & \multicolumn{3}{c}{0.25} & \multicolumn{3}{c}{0.5} & \multicolumn{3}{c}{0.75} & \multicolumn{3}{c}{1} \\ \cmidrule(lr){2-4}\cmidrule(lr){5-7}\cmidrule(lr){8-10}\cmidrule(lr){11-13}\cmidrule(lr){14-16}
  & $\Omega_1$ & $\Omega_{123}$ & $\Omega_{1|23}$ & $\Omega_1$ & $\Omega_{123}$ & $\Omega_{1|23}$ & $\Omega_1$ & $\Omega_{123}$ & $\Omega_{1|23}$ & $\Omega_1$ & $\Omega_{123}$ & $\Omega_{1|23}$ & $\Omega_1$ & $\Omega_{123}$ & \multicolumn{1}{c}{$\Omega_{1|23}$} \\ 
\midrule
$(b_{01}, b_{02}) = (0, 0.25)$  & $0.824$ & $0.429$ & $0.724$ & $0.852$ & $0.456$ & $0.757$ & $0.855$ & $0.455$ & $0.758$ & $0.858$ & $0.458$ & $0.757$ & $0.855$ & $0.454$ & $0.754$ \\
$(b_{01}, b_{02}) = (0, 0.5)$  & $0.825$ & $0.433$ & $0.733$ & $0.854$ & $0.454$ & $0.758$ & $0.858$ & $0.452$ & $0.751$ & $0.857$ & $0.454$ & $0.757$ & $0.858$ & $0.455$ & $0.759$ \\
$(b_{01}, b_{02}) = (0, 0.75)$  & $0.820$ & $0.430$ & $0.725$ & $0.857$ & $0.456$ & $0.760$ & $0.859$ & $0.453$ & $0.757$ & $0.860$ & $0.462$ & $0.758$ & $0.854$ & $0.456$ & $0.757$ \\
$(b_{01}, b_{02}) = (0, 1)$  & $0.825$ & $0.432$ & $0.730$ & $0.851$ & $0.455$ & $0.752$ & $0.853$ & $0.453$ & $0.756$ & $0.855$ & $0.453$ & $0.752$ & $0.857$ & $0.455$ & $0.756$ \\
$(b_{01}, b_{02}) = (0.25, 0.25)$  & $0.824$ & $0.433$ & $0.728$ & $0.858$ & $0.457$ & $0.760$ & $0.855$ & $0.454$ & $0.755$ & $0.854$ & $0.453$ & $0.756$ & $0.858$ & $0.459$ & $0.760$ \\
$(b_{01}, b_{02}) = (0.25, 0.5)$  & $0.827$ & $0.434$ & $0.726$ & $0.854$ & $0.456$ & $0.757$ & $0.853$ & $0.445$ & $0.751$ & $0.853$ & $0.449$ & $0.755$ & $0.853$ & $0.455$ & $0.754$ \\
$(b_{01}, b_{02}) = (0.25, 0.75)$  & $0.821$ & $0.433$ & $0.725$ & $0.859$ & $0.450$ & $0.760$ & $0.849$ & $0.450$ & $0.749$ & $0.858$ & $0.457$ & $0.757$ & $0.856$ & $0.454$ & $0.759$ \\
$(b_{01}, b_{02}) = (0.25, 1)$  & $0.826$ & $0.426$ & $0.727$ & $0.854$ & $0.449$ & $0.758$ & $0.853$ & $0.457$ & $0.755$ & $0.854$ & $0.449$ & $0.754$ & $0.854$ & $0.451$ & $0.754$ \\
$(b_{01}, b_{02}) = (0.5, 0.25)$  & $0.823$ & $0.432$ & $0.726$ & $0.854$ & $0.451$ & $0.752$ & $0.855$ & $0.456$ & $0.754$ & $0.854$ & $0.451$ & $0.752$ & $0.855$ & $0.452$ & $0.754$ \\
$(b_{01}, b_{02}) = (0.5, 0.5)$  & $0.828$ & $0.434$ & $0.731$ & $0.855$ & $0.452$ & $0.755$ & $0.857$ & $0.457$ & $0.758$ & $0.855$ & $0.449$ & $0.752$ & $0.855$ & $0.454$ & $0.755$ \\
$(b_{01}, b_{02}) = (0.5, 0.75)$  & $0.825$ & $0.430$ & $0.730$ & $0.854$ & $0.455$ & $0.752$ & $0.858$ & $0.453$ & $0.754$ & $0.857$ & $0.457$ & $0.761$ & $0.854$ & $0.453$ & $0.752$ \\
$(b_{01}, b_{02}) = (0.5, 1)$  & $0.823$ & $0.432$ & $0.723$ & $0.856$ & $0.456$ & $0.751$ & $0.857$ & $0.452$ & $0.757$ & $0.855$ & $0.452$ & $0.755$ & $0.859$ & $0.452$ & $0.753$ \\
$(b_{01}, b_{02}) = (0.75, 0.25)$  & $0.829$ & $0.440$ & $0.727$ & $0.857$ & $0.452$ & $0.758$ & $0.855$ & $0.453$ & $0.752$ & $0.857$ & $0.450$ & $0.758$ & $0.854$ & $0.448$ & $0.751$ \\
$(b_{01}, b_{02}) = (0.75, 0.5)$  & $0.831$ & $0.441$ & $0.733$ & $0.856$ & $0.449$ & $0.758$ & $0.857$ & $0.453$ & $0.758$ & $0.854$ & $0.450$ & $0.753$ & $0.853$ & $0.446$ & $0.754$ \\
$(b_{01}, b_{02}) = (0.75, 0.75)$  & $0.820$ & $0.428$ & $0.727$ & $0.856$ & $0.450$ & $0.752$ & $0.857$ & $0.452$ & $0.754$ & $0.858$ & $0.449$ & $0.758$ & $0.859$ & $0.450$ & $0.753$ \\
$(b_{01}, b_{02}) = (0.75, 1)$  & $0.823$ & $0.430$ & $0.729$ & $0.858$ & $0.451$ & $0.756$ & $0.858$ & $0.454$ & $0.755$ & $0.857$ & $0.454$ & $0.755$ & $0.852$ & $0.451$ & $0.750$ \\
$(b_{01}, b_{02}) = (1, 0.25)$  & $0.821$ & $0.433$ & $0.725$ & $0.853$ & $0.447$ & $0.752$ & $0.856$ & $0.454$ & $0.757$ & $0.856$ & $0.448$ & $0.755$ & $0.856$ & $0.451$ & $0.750$ \\
$(b_{01}, b_{02}) = (1, 0.5)$  & $0.830$ & $0.434$ & $0.735$ & $0.857$ & $0.452$ & $0.758$ & $0.855$ & $0.452$ & $0.759$ & $0.857$ & $0.453$ & $0.758$ & $0.853$ & $0.448$ & $0.757$ \\
$(b_{01}, b_{02}) = (1, 0.75)$  & $0.826$ & $0.436$ & $0.731$ & $0.855$ & $0.450$ & $0.756$ & $0.852$ & $0.449$ & $0.745$ & $0.853$ & $0.451$ & $0.752$ & $0.851$ & $0.443$ & $0.748$ \\
$(b_{01}, b_{02}) = (1, 1)$  & $0.822$ & $0.426$ & $0.729$ & $0.856$ & $0.453$ & $0.754$ & $0.855$ & $0.451$ & $0.757$ & $0.854$ & $0.448$ & $0.752$ & $0.851$ & $0.445$ & $0.750$ \\
\bottomrule 
\end{tabular}